\documentclass[10pt,twocolumn,twoside] {IEEEtran}

\makeatletter
\def\ps@headings{%
\def\@oddhead{\mbox{}\scriptsize\rightmark \hfil \thepage}%
\def\@evenhead{\scriptsize\thepage \hfil \leftmark\mbox{}}%
\def\@oddfoot{}%
\def\@evenfoot{}}
\makeatother 

\IEEEoverridecommandlockouts

\usepackage{amsfonts}
\usepackage[dvips]{graphicx}
\usepackage{times}
\usepackage{cite}
\usepackage{amsmath}
\usepackage{array}
\usepackage{amssymb}

\usepackage{stfloats}
\usepackage{slashbox}
\usepackage{graphicx}
\usepackage{footnote}
\usepackage{booktabs}
\usepackage{array}
\usepackage{algorithmic}
\usepackage{algorithm}
\usepackage{subeqnarray}
\usepackage{cases}
\usepackage{threeparttable}
\usepackage{color}
\usepackage{hyperref}
\usepackage{bbm}
\usepackage{CJK}
\usepackage{bm}
\usepackage{epstopdf}
\usepackage{multirow}
\usepackage{bm}
\usepackage{caption}
\usepackage{subcaption}
\usepackage{booktabs}
\usepackage{enumerate}
\usepackage{adjustbox}

\newcommand{\tabincell}[2]{\begin{tabular}{@{}#1@{}}#2\end{tabular}}

\usepackage{mathtools}

\DeclarePairedDelimiter\floor{\lfloor}{\rfloor}

\newtheorem{theorem}{Theorem}
\newtheorem{remark}{Remark}

\newtheorem{lemma}{Lemma}
\newtheorem{definition}{Definition}
\newtheorem{problem}{Problem}

\begin{document}

\title{Stochastic Throughput Optimization for Two-hop Systems with Finite Relay Buffers}

\author{Bo~Zhou, Ying~Cui,~\IEEEmembership{Member,~IEEE}, and Meixia~Tao,~\IEEEmembership{Senior~Member,~IEEE}
\thanks{This paper was presented in part at  IEEE Globecom 2014.

B.~Zhou, Y.~Cui and M.~Tao are with the Department of
Electronic Engineering at Shanghai Jiao Tong University, Shanghai,
200240, P. R. China. Email: \{b.zhou, cuiying, mxtao\}@sjtu.edu.cn.}
}
 \maketitle

\begin{abstract}
Optimal queueing control of multi-hop networks remains a challenging problem even in the simplest scenarios. In this paper, we consider a two-hop half-duplex relaying system with random channel connectivity. The relay is equipped with a finite buffer. We focus on stochastic link selection and transmission rate control to maximize the average system throughput subject to a half-duplex constraint. We formulate this stochastic optimization problem as an infinite horizon average cost Markov decision process (MDP), which is well-known to be a difficult problem. By using sample-path analysis and exploiting the specific problem structure, we first obtain an \emph{equivalent Bellman equation} with reduced state and action spaces. By using \emph{relative value iteration algorithm}, we analyze the properties of the value function of the MDP. Then, we show that the optimal policy has a threshold-based structure by characterizing the \emph{supermodularity} in the optimal control. Based on the threshold-based structure and Markov chain theory, we further simplify the original complex stochastic optimization problem to a static optimization problem over a small discrete feasible set and propose a low-complexity algorithm to solve the simplified static optimization problem by making use of its special structure. Furthermore,  we  obtain the closed-form optimal threshold for the symmetric case. The analytical results obtained in this paper also  provide design insights for two-hop relaying systems with  multiple relays equipped with finite relay buffers.
\end{abstract}
\section{Introduction}
The demand for communication services has been changing from traditional voice telephony services to mixed voice, data, and multimedia services.
When data and realtime services are considered, it is necessary to jointly consider  both physical layer issuers such as coding and modulation as well as higher layer issues such as network congestion and delay. It is also important to  model these services using queueing concepts\cite{queueing,surveyit}. On the other hand, to meet the explosive demand for these services, relaying has been shown effective for providing higher wireless date rate and better quality of service. Therefore, relay has been included in LTE-A \cite{ltea} and WiMAX \cite{wimax} standards, where both technologies support fixed and two-hop relays \cite{imt}.

Consider a two-hop relaying system with one source node (S), one half-duplex relay node (R) and one destination node (D) under i.i.d. on-off fading.
Under conventional \emph{decode-and-forward} (DF) relay protocol, a scheduling slot is divided into two transmission phases, i.e., the \emph{listening phase} (S-R) and the \emph{retransmission phase} (R-D). The S-R phase must be followed by the R-D phase\cite{cooperative}. Under \emph{the instantaneous flow balance constraint}, the system throughput is the minimum of the throughput from S to R and from R to D. Therefore, under random link connectivity, to achieve a non-zero system throughput (from S to D) within a scheduling slot, both the S-R and R-D links should be connected\cite{TIT15Cui}.

Now consider a finite buffer at R and apply cross-layer \emph{buffered decode-and-forward} (BDF) protocol to exploit the random channel connectivity and queueing\cite{TIT15Cui}. Under BDF, due to buffering at R, a scheduling slot can be adaptively allocated for the S-R transmission or the R-D transmission, according to the R queue length and link quality. Then, the throughput to D can be made non-zero provided that the R-D link is connected. While the buffer at R appears to offer obvious advantages, it is not clear how to design the optimal control to maximize the average system throughput given a finite relay buffer. Buffering a certain amount of bits at R can capture  R-D transmission opportunity (when only the R-D link is on) and  improve the  throughput in the future. However, buffering too many bits at R may waste  S-R transmission opportunity (when only the S-R link is on)  due to R buffer overflow. Therefore, it remains unclear how to take full advantage of the finite buffer at R to balance the transmission rates of the S-R and R-D links so as to maximize the average system throughput.

Recently, the idea of cross-layer design using queueing concepts has been considered in the context of multi-hop networks with buffers. In \cite{TIT15Cui} and \cite{ISIT11Cui}, the authors consider the delay-optimal control for two-hop networks with infinite buffers at the source and relay. Specifically, in \cite{ISIT11Cui}, the authors obtain a delay-optimal link selection policy for non-fading channels. Then, in\cite{TIT15Cui}, the authors extend the analysis to i.i.d. on/off fading channels and show that a threshold-based link selection policy is asymptotically delay-optimal when the scheduling slot duration tends to zero. However, it is not known whether the delay-optimal policy is still of a threshold-based structure.
In \cite{ztit}, the authors consider a two-hop relaying system with an infinite backlog at the source and an infinite buffer at the relay. The optimal link selection policies are obtained to maximize the average system throughput.
In the aforementioned references, the relay is assumed to be equipped with an infinite buffer and the proposed algorithms cannot guarantee that the instantaneous relay queue length is below a certain threshold. However, in practical systems, buffers are finite. The optimal designs for systems with infinite buffers do not necessarily lead to good performance for systems with finite buffers. In addition, in several practical networks, such as wireless sensor networks, wireless body area networks and wireless networks-on-chip, buffer size is limited. This is because that using buffers of large size would introduce practical issues, such as larger on-chip board space, increased memory-access latency and higher power consumption \cite{baron09,finiteline}.
Therefore, it is very important to consider finite relay buffers in designing optimal resource controls for multi-hop networks to support data and realtime services
\cite{ztwc2,finitebuffer,Long,xuewiopt,xuetvt}.

Lyapunov drift approach represents a systematic way to queue stabilization problems for general multi-hop networks with infinite buffers\cite{georgiadis2006,neely10}.
Specifically, Lyapunov drift approach mainly relies on quadratic Lyapunov functions, and can be used to obtain stochastic control algorithms with achieved utilities that are arbitrarily close to optimal.
The derived  control algorithms usually do not require system statistics predict beforehand and can be easily implemented online.
However, the traditional Lyapunov drift approach cannot properly handle systems with finite buffers. References \cite{finitebuffer} and \cite{Long} extend the traditional Lyapunov drift approach in \cite{georgiadis2006} and \cite{neely10} to design stochastic control algorithms for  multi-hop networks with infinite source buffers and finite relay buffers. In particular, \cite{finitebuffer} and \cite{Long} employ a new type of Lyapunov functions by multiplying queue backlogs of infinite buffers to the quadratic term of queue backlogs of finite buffers. Specifically, in \cite{finitebuffer}, the authors propose scheduling algorithms to stabilize source queues under a fixed routing design. In \cite{Long}, the authors propose joint flow control, routing and scheduling algorithms to maximize the throughput. References \cite{xuewiopt} and\cite{xuetvt} adopt similar approaches to those in \cite{finitebuffer} and \cite{Long}, and design control algorithms to optimize network utilities for multi-hop networks with finite source and relay buffers. However, the gap between the utility of each algorithm proposed in \cite{Long,xuetvt,xuewiopt}  and the optimal utility is inversely proportional to the buffer size. In other words, for the finite buffer case, the performance gap is always positive. Therefore, in contrast to the algorithms for the infinite buffer case in \cite{georgiadis2006} and \cite{neely10},  the algorithms for the finite buffer case in \cite{Long,xuetvt,xuewiopt} cannot achieve utilities that are arbitrarily close to optimal.

On the other hand, dynamic programming represents a systematic approach to optimal queueing control problems\cite{surveyit,yehphd,bertsekas}. Generally, there exist only numerical solutions, which do not typically offer many design insights and are usually impractical for implementation due to the curse of dimensionality \cite{bertsekas}.
For example, in \cite{wangtsp,wangtit}, the authors consider delay-aware control problems for two-hop relaying systems with multiple relay nodes  and propose suboptimal distributed numerical algorithms using approximate Markov Decision Process (MDP) and stochastic learning \cite{bertsekas}.
However, the obtained numerical algorithms may still be too complex for practical systems and do not offer many design insights.
Several existing works focus on characterizing structural properties of optimal policies to obtain design insights for simple queueing networks. However, most existing analytical results are for a single queue with either controlled arrival rate or departure rate  \cite{agarwal,ngo,ehsan,cdc}.
To the best of our knowledge, structural results for a single queue with both controlled arrival and departure rates are still unknown.
Furthermore, if the single queue has a finite buffer, the analytical results are limited. For example, \cite{cdc} characterizes structural properties of a single finite queue only for part of the queue state space.
 The challenge of structural analysis for finite-buffer systems stems from the \emph{reflection} effect (when a finite buffer is almost full)\cite{finiteline}.

In general, the stochastic throughput maximization for multi-hop systems with fading channels and finite relay buffers is still unknown even for the case of a simple two-hop relaying system.
In this paper, we shall tackle some of the technical challenges.
We consider a two-hop relaying system with one source node, one half-duplex relay node and one destination node  as well as random link connectivity. S has an infinite backlog and R is equipped with a finite buffer.
We consider stochastic link selection and transmission rate control to maximize the average system throughput subject to a half-duplex constraint. We formulate the stochastic average throughput optimization problem as an infinite horizon average cost MDP, which is well-known to be a difficult problem in general. By using sample-path analysis and exploiting the specific problem structure, we first obtain an \emph{equivalent Bellman equation} with reduced state and action spaces. By \emph{relative value iteration algorithm}, we analyze properties of the value function of the MDP. Then, based on these properties and the concept of \emph{supermodularity}, we show that the optimal policy has a threshold-based structure. By the structural properties of the optimal policy and Markov chain theory, we further simplify the original complex stochastic optimization problem to a  static optimization problem over a small discrete feasible set. We propose a low-complexity algorithm to solve the static optimization problem by making use of  its special structure.
Furthermore, we obtain the closed-form  optimal threshold for the symmetric case. Numerical results verify the theoretical analysis and demonstrate the performance gain of the derived optimal policy over the existing solutions.

\emph{Notations:} Boldface uppercase letters denote matrices and boldface lowercase letters denote vectors. $\mathbf{I}_n$ denotes an $n\times n$  identity matrix, the $k$-th column of which is denote as $\mathbf{e}_{k,n}$. $\mathbf{A}^{-1}$ and $\mathbf{A}^{T}$ denote the inverse and the transpose of matrix $\mathbf{A}$, respectively. $||\mathbf{q}||$ denotes the norm of vector $\mathbf{q}$. The important notations used in this paper are summarized in Table~\ref{tablenotation}.

\begin{table}[!htbp]
\footnotesize
\begin{adjustbox}{max width=0.49\textwidth}
\begin{tabular}{|c|c|}
\hline
$t$&slot index \\
\hline
$R_s, R_r$ & maximum transmission rates of S and R\\
\hline
$N_r$ & relay buffer size\\
\hline
$p_s, p_r$ & probabilities of ``ON'' for S-D and R-D links\\
\hline
$\mathbf{G}=(G_{s}, G_{r})$&joint CSI \\
\hline
$Q$&QSI \\
\hline
$\mathbf{X}=(Q,\mathbf{G})$&system state \\
\hline
$\bm{\mathcal{G}}=\mathcal{G}\times\mathcal{G}$&joint CSI state space\\
\hline
$\mathcal{Q}$&QSI state space \\
\hline
$\bm{\mathcal{X}}=\mathcal{Q}\times\bm{\mathcal{G}}$&system state space\\
\hline
$a_s, a_r$ & link selection actions for S-D and R-D links\\
\hline
$u_s, u_r$ & transmission rates of S and R\\
\hline
$\Omega=(\alpha, \mu)$ & link selection and transmission rate control policy\\
\hline
$V(Q)$ & value function\\
\hline
$J(Q,a_r)$ & state-action reward function\\
\hline
$\mathcal{C}$ & recurrent class of relay queue state process\\
\hline
$Q_{th}\in\mathcal{Q},q_{th}\in\mathcal{C}$ & threshold \\
\hline
\end{tabular}
\end{adjustbox}
\centering
\caption{\small{List of important notations}}\label{tablenotation}
\end{table}
\section{System Model}
As illustrated in Fig.~\ref{fig:systemmodel}, we consider a two-hop relaying system with one source node (S), one relay node (R) and one destination node (D). S cannot transmit packets to D due to the limited coverage and has to communicate with D with the help of R via the S-R link and the R-D link.\footnote{This two-hop relaying model can be used to model the Type 1 relay in LTE-Advanced and the non-transparent relay in WiMAX \cite{imt}.} R is half-duplex and equipped with a finite buffer.
We consider a discrete-time system, in which the time axis is partitioned into scheduling slots with unit slot duration. The slots are indexed by $t$ $(t=1,2,...)$.
\begin{figure}[h]
\begin{centering}
\includegraphics[scale=.3]{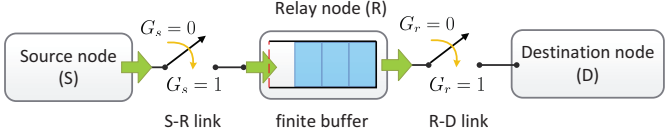}
 \caption{\small{System model.}}\label{fig:systemmodel}
\end{centering}
\end{figure}
\subsection{Physical Layer Model}
We model the channel fading of the S-R link and the R-D link  with i.i.d. random link connectivity.\footnote{This channel fading model is widely used in the literature \cite{onoff07,onoff12}.} Let $G_{s,t}, G_{r,t}\in\mathcal{G}\triangleq\{0,1\}$ denote the link connectivity state information (CSI) of the S-R link and the R-D link at slot $t$, respectively, where 1 denotes connected and 0 not connected. Let $\mathbf{G}_t\triangleq(G_{s,t}, G_{r,t})\in\bm{\mathcal{G}}\triangleq\mathcal{G}\times\mathcal{G}$ denote the joint CSI at the $t$-th slot,
where $\bm{\mathcal{G}}$ denotes the joint CSI state space.

\emph{Assumption 1 (Random Link Connectivity Model):}   $\{G_{s,t}\}$ and $\{G_{r,t}\}$ are both i.i.d. over time, where in each slot $t$, the probabilities of being 1 for $G_{s,t}$ and $G_{r,t}$ are $p_s$ and $p_r$, respectively, i.e., $\Pr[G_{s,t}=1]=p_s$ and $\Pr[G_{r,t}=1]=p_r$. Furthermore, $\{G_{s,t}\}$ and $\{G_{r,t}\}$ are independent of each other.

We assume fixed transmission powers of S and R, and consider packet transmission. The maximum transmission rates (i.e., the maximum numbers of packets transmitted within a slot) of the S-R link when $G_{s,t}=1$ and the R-D link when $G_{r,t}=1$ are given by $R_s$ and $R_r$, respectively.
Note that, to avoid the overflow of the finite R buffer, the actual transmission rate of S may be smaller than $R_s$.
In addition, the actual transmission rate of R may be smaller than $R_r$, subject to the availability of packets in the R buffer.
These will be further illustrated in Section III-A.

\subsection{Queueing Model}
We assume that S has an infinite backlog (i.e., always has data to transmit) and consider a finite buffer of size $N_r<\infty$ (in number of packets) at R. Note that $N_r$ can be arbitrarily large. Assume $N_r>\max\{R_s, R_r\}$. The finite buffer at R is used to hold the packet flow from S. We consider the \emph{buffered decode-and-forward} (BDF) protocol \cite{TIT15Cui} to exploit the potential benefit of buffering at R under random channel connectivity. Specifically, according to BDF, (i) S can transmit packets to R when the S-R link is connected, and R decodes and stores the packets from S in its buffer; (ii) R can transmit the packets in its buffer to D when the R-D link is connected. Using the buffer at R and BDF, we can dynamically select the S-R link or the R-D link to transmit and choose the corresponding transmission rate at each slot based on the channel fading and queue states, according to a link selection and transmission rate control policy defined in Section III-A.

Therefore, as illustrated in Fig.~\ref{fig:systemmodel}, the simple two-hop relaying system with on/off channel connectivity can be modeled  as a single queue with controlled arrival rate and departure rate. Let $Q_t\in\mathcal{Q}$ denote the queue state information (QSI) (in number of packets) at the R buffer at the beginning of the $t$-th slot, where $\mathcal{Q}\triangleq\{0,1,\cdots,N_r\}$ denotes the QSI state space. The queue dynamics under the control policy will be illustrated in Section III-B.

\section{Problem Formulation}
\subsection{Control Policy}
For notation convenience, we denote $\mathbf{X}_t\triangleq (Q_t,\mathbf{G}_t)\in\bm{\mathcal{X}\triangleq}\mathcal{Q}\times\bm{\mathcal{G}}$ as the \emph{system state} at the $t$-th slot, where $\bm{\mathcal{X}}$ denotes the system state space. Let $a_{s,t}\in\{0,1\}$ and $a_{r,t}\in\{0,1\}$ denote whether the S-R link or the R-D link is scheduled, respectively, in the $t$-th slot, where 1 denotes scheduled and 0 otherwise. Let $u_{s,t}\in\{0,1,\cdots, R_s\}$ and $u_{r,t}\in\{0,1,\cdots, R_r\}$ denote the transmission rates of S and R in the $t$-th slot, respectively. Given an observed system state $\mathbf{X}$, the link selection action $ (a_s,a_r)\in\{0,1\}^2$ and the transmission rate control action $(u_s,u_r)\in\{0,1,\cdots, R_s\}\times\{0,1,\cdots, R_r\}$ are determined according to a stationary  policy defined below.

\begin{definition}[\text{Stationary Policy}]
  A stationary link selection and transmission rate control policy $\Omega\triangleq(\alpha,\mu)$ is a mapping from the system state $\mathbf{X}\triangleq(Q,\mathbf{G})$ to the link selection action $(a_s,a_r)$ and the transmission rate control action $(u_s,u_r)$, where $\alpha(\mathbf{X})=(a_s,a_r)$ and $\mu(\mathbf{X})=(u_s,u_r)$ satisfy the following constraints:
\begin{enumerate}
  \item $a_s,a_r\in\{0,1\}$;
  \item $a_s+a_r\leq 1$ (orthogonal link selection);
  \item  \noindent$(a_s,a_r)=\begin{cases}(0,0),  & \mathbf{G}=(0,0) \\
            (0,1),  &\mathbf{G}=(0,1) \\
            (1,0),  &  \mathbf{G}=(1,0)

  \end{cases}$\\(at least one link is not connected);
  \item $u_s\in\{0,1,\cdots,\min\{R_s, N_r-Q\}\}$  (departure rate at S);

  \item $u_r\in\{0,1,\cdots,\min\{R_r,Q\}\}$ (departure rate at R).
\end{enumerate}
\label{definition:definition1}
\end{definition}

Note that, our focus for the link selection control is  on the design of $(a_s,a_r)$ when  $\mathbf{G}=(1,1)$. Moreover,
the departure rates (actual transmission rates) of S and R, i.e., $u_s$ and $u_r$, may be smaller than $R_s$ and $R_r$, respectively, due to the following reasons.
When the finite R buffer does not have enough space, to avoid buffer overflow and the resulting packet loss,  $u_s$ is smaller than $R_s$. When the R buffer  does not have enough packets to transmit,  $u_r$ is smaller than $R_r$. Thus, we have the constraints in 4) and 5).

\subsection{MDP Formulation}
Given a stationary control policy $\Omega$ defined in Definition~\ref{definition:definition1}, the queue dynamics at R is given by:
\begin{equation}
  Q_{t+1}=Q_t+a_{s,t}u_{s,t}-a_{r,t}u_{r,t},~\forall t=1,2,\cdots.\label{eqn:queue-dyn-t}
\end{equation}
From Assumption 1 and the queue dynamics in \eqref{eqn:queue-dyn-t}, we can see that the induced random process $\{\mathbf{X}_t\}$ under policy $\Omega$ is a Markov chain with the following transition probability
\begin{equation}\label{eqn:prob}
\Pr[\mathbf{X}_{t+1}|\mathbf{X}_t,\Omega(\mathbf{X}_t)]=\Pr[\mathbf{G}_{t+1}]\Pr[Q_{t+1}|\mathbf{X}_t,\Omega(\mathbf{X}_t)].
\end{equation}

In this paper, we restrict our attention to stationary unichain policies.\footnote{A unichain policy is a policy, under which the induced Markov chain has a single recurrent class (and possibly some transient states)\cite{bertsekas}.} For a given stationary unchain policy $\Omega$, the average system throughput is given by:
\begin{equation}\label{eqn:Rpi}
  \bar{R}^\Omega\triangleq\liminf_{T\to\infty}\frac{1}{T}\sum_{t=1}^T \mathbb{E}\left[r(\mathbf{X}_t,\Omega(\mathbf{X}_t))\right],
\end{equation}
where $r(\mathbf{X}_t,\Omega(\mathbf{X}_t))\triangleq a_{r,t}u_{r,t}$ is the per-stage reward (i.e., the departure rate at R at slot $t$, indicating the number of packets delivered by the two-hop relaying system) and the expectation is taken w.r.t. the measure induced by policy $\Omega$.

 We wish to find an optimal link selection and transmission rate control policy $\Omega^*$ to maximize the average system throughput $\bar{R}^\Omega$ in \eqref{eqn:Rpi}.\footnote{By Little's law, maximizing the average system throughput in Problem~\ref{problem:originalproblem} is equivalent to minimizing the upper bound of the average delay in the relay with a finite buffer.}
\begin{problem}[Stochastic Throughput Optimization]
  \begin{equation}
  \bar{R}^*\triangleq \max_{\Omega}\liminf_{T\to\infty}\frac{1}{T}\sum_{t=1}^T \mathbb{E}\left[r(\mathbf{X}_t,\Omega(\mathbf{X}_t))\right],\label{eqn:problem1}
\end{equation}\label{problem:originalproblem}
where $\Omega$ is a stationary unchain policy satisfying the constraints in Definition~\ref{definition:definition1}.
\end{problem}
Please note that, in Problem~\ref{problem:originalproblem}, we assume the existence of a stationary unichain policy achieving the maximum in \eqref{eqn:problem1}. Latter, in Theorem~\ref{theorem:theorem1}, we shall prove the existence of such a policy.
Problem~\ref{problem:originalproblem} is an infinite horizon average cost MDP, which is well-known to be a difficult problem \cite{bertsekas}. While dynamic programming represents a systematic approach for MDPs, there generally exist only numerical solutions, which do not typically offer many design insights, and are usually not practical due to the curse of dimensionality\cite{bertsekas}.

Fig. \ref{fig:flow} illustrates in the remainder of this paper, how we shall address the above challenges to solve Problem~\ref{problem:originalproblem}. Specifically, in Sections IV and V, we shall analyze the  properties of the optimal policy. Based on these properties, we shall simplify Problem~\ref{problem:originalproblem} to a static optimization problem (Problem~\ref{problem:equivalentproblem})  and develop a low-complexity algorithm (Algorithm~\ref{alg:total}) to solve it. Finally,  we shall obtain the corresponding static optimization problem (Problem~\ref{problem:problemSymm}) for the symmetric case and derive its closed-form optimal solution.

\begin{figure}[t]
\begin{centering}
\includegraphics[scale=0.38]{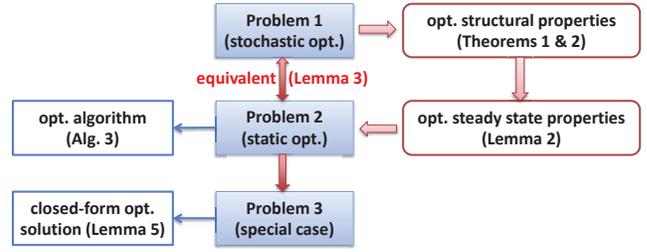}
 \caption{\small{Proposed solution to Problem~\ref{problem:originalproblem}.}}\label{fig:flow}
\end{centering}
\end{figure}

\section{Structure of Optimal Policy}
In this section, we first obtain an equivalent Bellman equation based on reduced state and action spaces. Then, we show that the optimal policy has a threshold-based structure.
\subsection{Optimality Equation}
By exploiting some special structures in our problem, we obtain the following equivalent Bellman equation by reducing the state and action spaces. By solving the Bellman equation, we can obtain the optimal policy to Problem~\ref{problem:originalproblem}.
\begin{theorem}[Equivalent Bellman Equation]
 (i) The optimal transmission rate control policy $\mu^*$  is given by:
\begin{equation}\label{eqn:rate}
\mu^*(\mathbf{X})=\left(\min\{R_s, N_r-Q\},\min\{R_r,Q\}\right), \forall \mathbf{X} \in\bm{\mathcal{X}}.
\end{equation}

(ii) There exists $(\theta, \{V(Q)\})$ satisfying  the following \emph{equivalent Bellman equation}:
\begin{align}\label{eqn:reduceAction}
&\theta+V(Q)=\bar{p}_s\bar{p}_rV(Q)+p_s\bar{p}_rV(\min\{Q+R_s,N_r\})\nonumber\\
&+\bar{p}_sp_r\left(\min\{Q,R_r\}+V([Q-R_r]^+)\right)\nonumber\\
&+p_sp_r\max_{a_r}\Big\{a_r\min\{Q,R_r\}\nonumber\\&+V\big(Q+\bar{a}_r\min\{N_r-Q,R_s\}-a_r\min\{Q,R_r\}\big)\Big\},\nonumber\\
&\hspace{64mm}\forall Q\in\mathcal{Q},
\end{align}
where $[x]^+\triangleq\max\{x,0\}$, $\bar{p}_s\triangleq 1-p_s$, $\bar{p}_r\triangleq 1-p_r$ and  $\bar{a}_r\triangleq 1-a_r$. $\theta=\bar{R}^*$ is the optimal value to Problem~\ref{problem:originalproblem} for all initial state $\mathbf{X}_1\in\bm{\mathcal{X}}$ and $V(\cdot )$ is called the value function.

(iii)
The optimal link selection  policy $\alpha^*$ is given by:
\begin{equation}\label{eqn:reducealpha}
\alpha^*(\mathbf{X})=\left\{
          \begin{array}{ll}
            (0,0), & \hbox{$\mathbf{G}=(0,0)$;} \\
            (0,1), & \hbox{$\mathbf{G}=(0,1)$;} \\
            (1,0), & \hbox{$\mathbf{G}=(1,0)$;} \\
            \left(\bar\alpha_r^*(Q),\alpha_r^*(Q)\right) & \hbox{$\mathbf{G}=(1,1)$.}
          \end{array} \quad \forall \mathbf{X} \in\bm{\mathcal{X}},
        \right.
\end{equation}
where 
\begin{align}
  &\alpha_r^*(Q)\triangleq \arg\max_{a_r}\Big\{a_r\min\{Q,R_r\}\nonumber\\
  &+V\big(Q+\bar{a}_r\min\{N_r-Q,R_s\}-a_r\min\{Q,R_r\}\big)\Big\}, \nonumber\\
  & \hspace{64mm} \forall Q\in\mathcal{Q},\label{eqn:omegaalpha}
\end{align}
and $\bar\alpha_r^*(Q)\triangleq 1-\alpha_r^*(Q)$.
\label{theorem:theorem1}
\end{theorem}
\begin{proof}
  Please see Appendix A.
\end{proof}

Note that, the four terms in the R.H.S of \eqref{eqn:reduceAction} correspond to the per-stage reward plus the value function of the updated queue state for $\mathbf{G}=(0,0), (1,0), (0,1)$ and $(1,1)$, respectively, under the optimal transmission rate control policy in \eqref{eqn:rate}, the link selection policy in Definition 1 for $\mathbf{G}=(0,0), (1,0)$ and $(0,1)$, and the optimal link selection policy for $\mathbf{G}=(1,1)$. Therefore, the R.H.S of \eqref{eqn:reduceAction} indicates the expectation of  the per-stage reward plus the value function of the updated queue state under the optimal policy, where the expectation is over the channel state $\mathbf{G}$.
\begin{remark}[Reduction of State and Action Spaces]
The Bellman equation in \eqref{eqn:reduceAction} is defined over the QSI state space $\mathcal{Q}$. Thus, the system state space $\mathcal{Q}\times\bm{\mathcal{G}}$ in Definition~\ref{definition:definition1}    is reduced to the QSI state space $\mathcal{Q}$. The action space reduction  can be observed by comparing Definition~\ref{definition:definition1} with   \eqref{eqn:rate} and \eqref{eqn:reducealpha}.
\label{remark:remarkofTheorem1}
\end{remark}

Note that  the closed-form optimal transmission rate control policy ${\mu}^*$  has already been obtained in \eqref{eqn:rate}. The optimal link selection policy ${\alpha}^*$ is determined by the policy $\alpha_r^*$ in \eqref{eqn:omegaalpha}. Thus, we only need to consider the optimal link selection for $\mathbf{G}=(1,1)$.  In the following, we also refer to $\alpha_r^*$ as the optimal link selection policy.
To obtain the optimal policy, it remains to characterize $\alpha_r^*$. From Theorem~\ref{theorem:theorem1}, we can see that $\alpha_r^*$ depends on the QSI state $Q$ through the value function $V(\cdot)$. Obtaining $V (\cdot)$ involves solving the equivalent Bellman equation in \eqref{eqn:reduceAction} for all $Q$. There is no closed-form solution in general\cite{bertsekas}. Brute force solutions such as value iteration and policy iteration are usually impractical for implementation and do not yield many design insights \cite{bertsekas}. Therefore, it is desirable to study the structure of $\alpha_r^*$.

\subsection{Threshold Structure of Optimal Link Selection Policy}
To further simplify the problem and obtain design insights, we study the structure of the optimal link selection policy.
In the existing literature, structural properties of optimal policies are characterized for simple networks by studying properties of the value function.
  For example, most existing works consider the structural analysis of a single queue with either controlled arrival or departure rates \cite{agarwal,ngo,ehsan,cdc}.  However, we control both the arrival and departure rates of the relay queue.
 Moreover, we consider a finite buffer, which has \emph{reflection} effect when the buffer is almost full \cite{finiteline}, and general system parameters, i.e., $R_s, R_r$ and $N_r$.
Therefore, it is more challenging to explore the properties of the value function in our system.

First, by  the \emph{relative value iteration algorithm (RVIA)}\footnote{RVIA is a commonly used numerical method for iteratively computing the value function, which is the solution to the Bellman equation for the infinite horizon average cost MDP \cite[Chapter 4.3]{bertsekas}. The details of RVIA can be found in Appendix B.}, we can iteratively prove the following properties of the value function.
\begin{lemma}[Properties of Value Function]
      The value function $V(Q)$ satisfies the following properties:
    \begin{enumerate}
      \item $V(Q)$ is monotonically non-decreasing in $Q$;\label{eqn:property1}
      \item $V(Q+1)-V(Q)\leq 1$, $Q\in\{0,1,\cdots,N_r-1\}$;\label{eqn:property2}
      \item $V(Q+R_s+R_r+1)-V(Q+R_s+R_r)\leq V(Q+1)-V(Q)$, $Q\in\{0,1,...,N_r-(R_s+R_r+1)\}$.\label{eqn:property3}
    \end{enumerate}
\label{lemma:propertiesofV}
\end{lemma}
\begin{proof}
  Please see Appendix B.
\end{proof}

\begin{remark}[Interpretation of Lemma~\ref{lemma:propertiesofV}]Property 1 generally holds for single-queue systems and is widely studied in the existing literature.
Property 2 results from the throughput maximization problem considered in this work.  This property does not hold for sum queue length minimization problems considered in most existing literature.
Property 3 indicates that $V(Q)$ is $K$-concave\footnote{A function $f(x)$: $\mathbb{Z}^+\cup\{0\} \rightarrow \mathbb{R}$ is $K$-concave (where $K\in\mathbb{Z}^+$) if   $f(x+K+1)-f(x+K)\leq f(x+1)-f(x)$.} with $K=R_s+R_r\geq2$. This stems from the relay queue with both controlled arrival and departure rates. In contrast, most existing works consider a single queue with either controlled arrival rate or departure rate, and the corresponding value function is 1-concave.
\end{remark}

Next, define the state-action reward function as follows\cite{ngo}
\begin{align}
  &J(Q,a_r )\triangleq\bar{p}_s\bar{p}_rV(Q)+p_s\bar{p}_rV(\min\{Q+R_s,N_r\})\nonumber\\
&+\bar{p}_sp_r\left(\min\{Q,R_r\}+V([Q-R_r]^+)\right)\nonumber\\&+p_sp_r\Big[a_r \min\{Q,R_r\}\nonumber\\&+V\big(Q+\bar{a}_r \min\{N_r-Q,R_s\}-a_r \min\{Q,R_r\}\big)\Big].\label{eqn:state_action_func}
\end{align}
Note that $J(Q,a_r)$ is related to the R.H.S. of the Bellman equation in \eqref{eqn:reduceAction}.
The R.H.S. of \eqref{eqn:state_action_func} indicates the expectation of the per-stage reward plus the value function of the updated queue state under the optimal transmission rate control policy  in \eqref{eqn:rate}, the link selection policy in Definition 1 for $\mathbf{G}=(0,0), (1,0)$ and $(0,1)$, and any link selection policy satisfying $a_s,a_r\in\{0,1\}$ and $a_s+a_r=1$ for $\mathbf{G}=(1,1)$.

By Lemma~\ref{lemma:propertiesofV} and \eqref{eqn:state_action_func}, we can show that the state-action reward function $J(Q,a_r )$ is supermodular\footnote{A function $f(x,y)$: $\mathbb{Z}^2\rightarrow \mathbb{R}$ is supermodular in  $(x,y)$ if
  $f(x+1,y+1)-f(x+1,y)\geq f(x,y+1)-f(x,y)$\cite{puterman}.} in $(Q,a_r )$, i.e.,
\begin{equation}\label{eqn:super}
  J(Q+1,1)-J(Q+1,0)\geq J(Q,1)-J(Q,0).
\end{equation}
By \cite[Lemma 4.7.1]{puterman}, supermodularity is a sufficient condition for the monotone policies to be optimal. Thus, we have the following theorem.
\begin{theorem}[Threshold Structure of Optimal Policy]
  There exists   $Q_{th}^*\in \mathcal Q$ such that the optimal link selection policy for $\mathbf{G}=(1,1)$ has the threshold-based structure, i.e.,
\begin{equation}\label{eqn:threshold}
  \alpha_r^*(Q)=\left\{
    \begin{array}{ll}
      1, & \hbox{if $Q>Q_{th}^*$;} \\
      0, & \hbox{otherwise.}
    \end{array}
  \right.
\end{equation}
$Q_{th}^*$ is the optimal threshold.
\label{theorem:theorem2}
\end{theorem}
\begin{proof}
Please see Appendix C.
\end{proof}

\begin{remark}[Interpretation of Theorem~\ref{theorem:theorem2}]
By Theorem~\ref{theorem:theorem2}, we know that when $\mathbf{G}=(1,1)$, it is optimal to schedule the R-D link if $Q>Q_{th}^*$ and to schedule the S-R link otherwise. The intuition is as follows. When the relay queue length is large ($Q>Q_{th}^*$), the S-R transmission opportunities may be wasted when $\mathbf{G}=(1,0)$ due to the overflow of the finite R buffer. Therefore, when $Q>Q_{th}^*$, we should reduce the relay queue length at $\mathbf{G}=(1,1)$. When the relay queue length is small ($Q\leq Q_{th}^*$), the R-D transmission opportunities may be wasted when $\mathbf{G}=(0,1)$, as there may not be enough packets left to transmit. Therefore, when  $Q\leq Q_{th}^*$, we should schedule the S-R link at $\mathbf{G}=(1,1)$. These design insights also hold for two-hop relaying systems with  multiple relays which are equipped with finite relay buffers.
 \label{remark:remarkofTheorem2}
\end{remark}

\section{Optimal Solution for General Case}
In this section, we first obtain a simplified static optimization problem for Problem~\ref{problem:originalproblem} by making use of the structural properties of the optimal policy in Theorems~\ref{theorem:theorem1} and \ref{theorem:theorem2}. Then, based on the special structure, we develop a low-complexity algorithm to solve the static optimization problem.
\subsection{Recurrent Class}\label{app:generalcaseA}

By the structural properties of the optimal policy in Theorems~\ref{theorem:theorem1} and~\ref{theorem:theorem2}, we can restrict our attention to the optimal transmission rate control in \eqref{eqn:rate} and a threshold-based link selection policy $\alpha_r$ for $\mathbf{G}=(1,1)$, i.e.,
\begin{equation}\label{eqn:thresholdlemma2}
  \alpha_r(Q)=\left\{
    \begin{array}{ll}
      1, & \hbox{if $Q>Q_{th}$;} \\
      0, & \hbox{otherwise.}
    \end{array}
  \right.
\end{equation}
 where  $Q_{th}\in\mathcal Q$ is the threshold.
In the following, we use $\{Q_t\}$ to denote the relay queue state process  under the policies in \eqref{eqn:rate} and \eqref{eqn:thresholdlemma2}. $\{Q_t\}$ is a stationary Discrete-Time Markov Chain (DTMC)\cite{gallager}, the transition probabilities of which are determined by the threshold $Q_{th}$ and the statistics of the CSI (i.e., $p_s$ and $p_r$).
Fig.~\ref{fig:transprob}  and  Fig.~\ref{fig:generalex} illustrate the transition from any state $Q\in\mathcal{Q}$ and the transition diagram for $\{Q_t\}$, respectively, under the optimal transmission rate control in \eqref{eqn:rate} and  the threshold-based link selection control in \eqref{eqn:thresholdlemma2}.

\begin{figure}[!h]
\begin{minipage}[ct]{.5\linewidth}
        \centering
        \includegraphics[scale=.44]{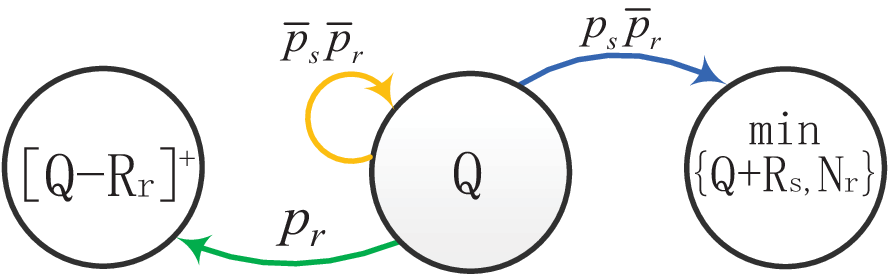}
        \centering
        \subcaption{\small{$Q>Q_{th}$}}
\end{minipage}%
\begin{minipage}[ct]{.5\linewidth}
\centering
        \includegraphics[scale=.44]{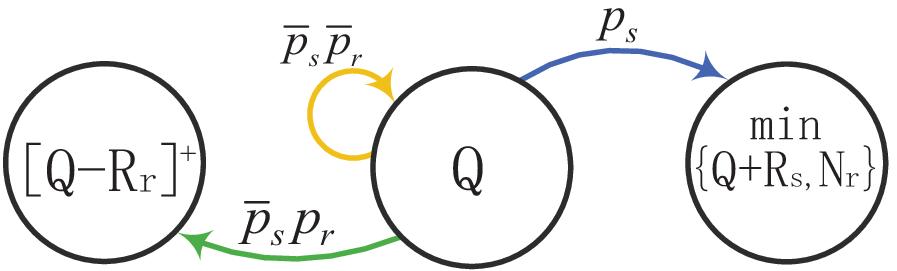}
         \subcaption{\small{$Q\leq Q_{th}$}}
\end{minipage}
        \caption{\small{Illustration of transitions from  state $Q\in\mathcal{Q}$.}}\label{fig:transprob}
\end{figure}
\begin{figure}[!h]
\begin{minipage}[ct]{.5\linewidth}

        \centering
        \includegraphics[scale=.38]{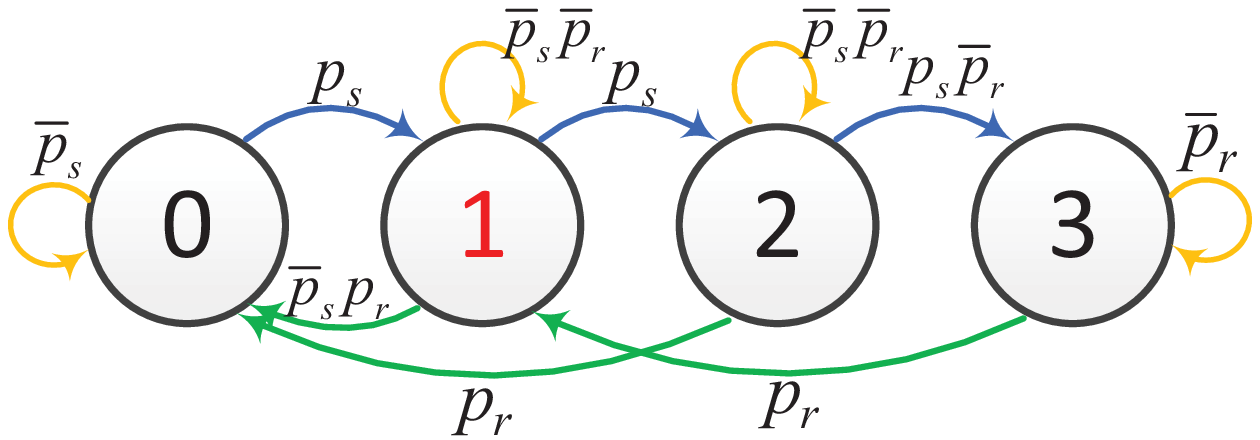}
        \centering
        \subcaption{\small{$l=0$, $R_s=2$, $R_r=1$.}}
\end{minipage}%
\begin{minipage}[ct]{.5\linewidth}
\centering
        \includegraphics[scale=.38]{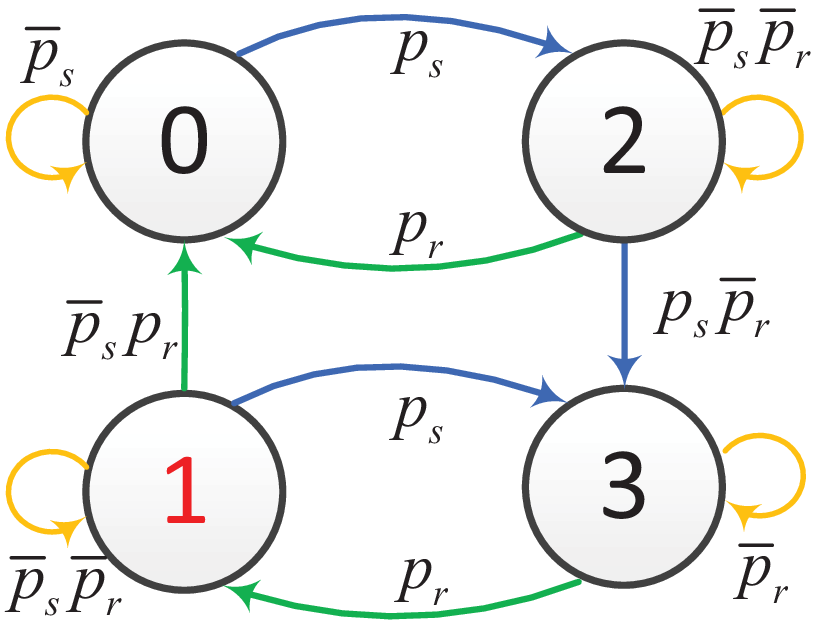}
         \subcaption{\small{$l\neq0$, $R_s=2$, $R_r=2$.}}
\end{minipage}
        \caption{\small{Illustration of the transition diagram of $\{Q_t\}$. $\mathcal{Q}=\{0,1,2,3\}$, $N_r=3$ and $Q_{th}=1$.}}\label{fig:generalex}
\end{figure}

Next, we study the steady-state probabilities of $\{Q_t\}$.
Note that, the steady-state probability of each transient state is zero \cite{gallager}, and hence the throughput of the transient states will not contribute to the ergodic throughput.
In other words, the ergodic system throughput is equal to the average throughput over the recurrent class of $\{Q_t\}$. Thus, to calculate the ergodic throughput, we first characterize the recurrent class of $\{Q_t\}$.
Let $R_s/R_r=a/b$ where $a$ and $b$ are two positive integers having no factors in common. Denote
\begin{equation}
  R\triangleq R_s/a~(=R_r/b)\label{eqn:R}.
\end{equation}
Since $N_r > \max\{R_s,R_r\}$, there exist $n\in\{1,\cdots,\floor*{\frac{N_r}{R}}\}$ and $l\in\{0,\cdots,R-1\}$ such that
\begin{equation}
  N_r=nR+l.\label{eqn:Nr}
\end{equation}
Using the B\'{e}zout's identity, we characterize the recurrent class $\mathcal{C}$ of $\{Q_t\}$ in the following lemma.
\begin{lemma}[Recurrent Class]
  For any $R_s, R_r,N_r$, under the optimal transmission rate control in \eqref{eqn:rate} and a threshold-based link selection policy in \eqref{eqn:thresholdlemma2} with any $Q_{th}\in\mathcal Q$, the recurrent class $\mathcal{C}$ of $\{Q_t\}$ is given by
\begin{equation*}\mathcal{C}=\left\{
                                   \begin{array}{ll}
                                      \{0,R,2R,\cdots,nR\}, & \hbox{if $l=0$} \\
                                      \{0,R,2R,\cdots,nR,l,l+R,l+2R,\cdots,N_r\}, & \hbox{if $l\neq 0$}
                                    \end{array}
                                  \right.
\end{equation*}
where $R$ is given by \eqref{eqn:R} and $l,n$ satisfy \eqref{eqn:Nr}. The size of $\mathcal{C}$ is $|\mathcal{C}|=(l+1)(n+1)$.
\label{lemma:recurrentclass}
\end{lemma}
\begin{proof}
  Please see Appendix D.
\end{proof}

Note that  $\mathcal{C}\subseteq \mathcal Q$ and $\mathcal{C}$ is the same for any $Q_{th}\in\mathcal Q$ $(|\mathcal{Q}|=N_r+1)$.

\subsection{Equivalent Problem}
We first consider a threshold-based policy in \eqref{eqn:thresholdlemma2} with the threshold chosen from $\mathcal{C}$ instead of $\mathcal Q$. Denote this threshold as $q_{th}$. We wish to find the optimal threshold $q_{th}^*\in\mathcal{C}$ to maximize the ergodic system throughput (i.e., the ergodic reward of $\{Q_t\}$).
Later, in Lemma~\ref{lemma:relationship12}, we shall show the relationship between $q_{th}^*\in\mathcal{C}$ and $Q_{th}^*\in\mathcal{Q}$.

As illustrated in Section V-A, we focus on the computation of the average throughput over the recurrent class $\mathcal{C}$.
Given $q_{th}$, we can express the transition probability  from $i$ to $j$ as $p_{i,j}(q_{th})$, where $i,j\in \mathcal{C}$. Let $\mathbf{P}(q_{th})\triangleq \left(p_{i,j}(q_{th})\right)_{i,j\in\mathcal{C}}$ and $\bm{\pi}(q_{th})\triangleq \left(\pi_i(q_{th})\right)_{i\in\mathcal{C}}$ denote the transition probability matrix and  the steady-state probability row vector of the recurrent class $\mathcal{C}$, respectively. Note that $\mathbf{P}(q_{th})$ is  fully determined by $q_{th}$ and the statistics of the CSI (i.e., $p_s$ and $p_r$), and can be easily obtained, as illustrated in Fig. \ref{fig:transprob}.
By the Perron-Frobenius theorem\cite{gallager}, $\bm{\pi}(q_{th})$ can be computed from the following system of linear equations:
\begin{align}
\begin{cases}
\bm{\pi}(q_{th})\mathbf{P}(q_{th})=\bm{\pi}(q_{th})\\
||\bm{\pi}(q_{th})||=1
\end{cases}.\label{eqn:comp-pi}
\end{align}
 Let $r_i(q_{th})$ denote the average departure rate at state $i\in\mathcal Q$ under the threshold $q_{th}$.
 According to the threshold-based link selection policy in \eqref{eqn:thresholdlemma2}, we know that: (i) if queue state $i> q_{th}$, the R-D link is selected when $\mathbf{G}=(0,1)$ or $\mathbf{G}=(1,1)$; (ii) if queue state $i\leq q_{th}$, the R-D link is selected only when $\mathbf{G}=(0,1)$. Thus, we have:
\begin{equation}
  r_i(q_{th})=\left\{
           \begin{array}{ll}
             p_r\min\{i,R_r\}, & \hbox{if $i > q_{th}$;} \\
             \bar{p}_sp_r\min\{i,R_r\}, & \hbox{otherwise.}
           \end{array}
         \right.\label{eqn:comp-ri}
\end{equation}
Let $\mathbf{r}(q_{th})\triangleq \left(r_i(q_{th})\right)_{i\in\mathcal{C}}$ denote the average departure rate column vector of the recurrent class $\mathcal{C}$. Therefore, the ergodic system throughput  can be expressed as $\bm{\pi}(q_{th})\mathbf{r}(q_{th})$.

Now, we formulate a static optimization problem to maximize the ergodic system throughput as below.
\begin{problem}[Equivalent Optimization Problem]
  \begin{align}
    \bar{r}^*\triangleq \max_{q_{th}\in\mathcal{C}}~\bm{\pi}(q_{th})\mathbf{r}(q_{th}).\label{eqn:problem2}
\end{align}
\label{problem:equivalentproblem}
\end{problem}
Note that, $q_{th}^*\in{\mathcal{C}}$ is the optimal solution to Problem~\ref{problem:equivalentproblem}.
and $Q_{th}^*\in\mathcal Q$ is the optimal threshold to Problem~\ref{problem:originalproblem}.
The following lemma summarizes the relationship between $q_{th}^*$ and  $Q_{th}^*$.
\begin{lemma}[Relationship between Problem~\ref{problem:originalproblem} and Problem~\ref{problem:equivalentproblem}]
  The optimal values to Problems 1 and 2 are the same, i.e., $\bar{R}^*=\bar{r}^*$. If $q^*_{th}=N_r$, then $Q_{th}^*=q^*_{th}$. If $q^*_{th}<N_r$, then any threshold $Q_{th}^*\in\{Q|q_{th}^*\leq Q<q_{th,next}^*,Q\in\mathcal{Q}\}$ is optimal to Problem 1, where $q_{th,next}^*\triangleq\min\{i|i>q_{th}^*,i\in\mathcal{C}\}$.\label{lemma:relationship12}
\end{lemma}\begin{proof}
Please see Appendix E.
\end{proof}

By Lemma~\ref{lemma:relationship12}, instead of solving Problem~\ref{problem:originalproblem}, which is a complex stochastic optimization problem, we can solve Problem~\ref{problem:equivalentproblem}, which is a static problem over the smaller feasible set $\mathcal{C}\subseteq \mathcal Q$.

\subsection{Algorithm for Problem~\ref{problem:equivalentproblem}}
Problem~\ref{problem:equivalentproblem} is a discrete optimization problem over the feasible set $\mathcal{C}$. It can be solved in a brute-force way by computing $\bm\pi(q_{th})$ for each $q_{th}\in\mathcal{C}$ separately. The brute-force method has high complexity and fails to exploit the structure of the problem.
In this part, we develop a low-complexity algorithm to solve Problem~\ref{problem:equivalentproblem} by computing $\bm\pi(q_{th})$ for all $q_{th}\in\mathcal{C}$ iteratively based on the special structure of $\mathbf{P}(q_{th})$.

 We sort the elements of $\mathcal{C}$ in ascending order, i.e., $c_1,c_2,\cdots,c_{|\mathcal{C}|}$, where $c_k$ denotes the $k$-th smallest element in $\mathcal{C}$. For notation simplicity, we use $\mathbf{P}(k)$  and $\bm{\pi}(k)$ to represent $\mathbf{P}(q_{th})$ and $\bm{\pi}(q_{th})$, respectively, where $c_k=q_{th}$. In other words,  each variable in $\mathcal{C}$ is indexed by $k$.
 Denote
 \begin{equation}
 \mathbf{A}(k)\triangleq \mathbf{I}_{|\mathcal{C}|}-\mathbf{P}(k)^T.\label{eqn:Ak}
 \end{equation}
Note that the size of $\mathbf{A}(k)$ is $|\mathcal{C}|\times |\mathcal{C}|$.
The system of linear equations in \eqref{eqn:comp-pi} can be transformed to the following system of linear equations:
\begin{align}
\begin{cases}
  \mathbf{A}(k)\bm{\pi}(k)=0\\
  ||\bm{\pi}(k)||=1
\end{cases}.\label{eqn:comp-pi-k}
\end{align}
The steady-state probability vector $\bm{\pi}(k)$ in \eqref{eqn:comp-pi-k} can be obtained using the partition factorization method\cite{computepi} as follows.
By removing the $(k+1)$-th column and the $|\mathcal{C}|$-th row of $\mathbf{A}(k)$, we obtain a submatrix of $\mathbf{A}(k)$, denoted as $\hat{\mathbf{A}}(k)$. Note that the size of $\hat{\mathbf{A}}(k)$ is $(|\mathcal{C}|-1)\times(|\mathcal{C}|-1)$.
Accordingly, let $\mathbf{K}(k)$ denote the $|\mathcal{C}|\times |\mathcal{C}|$ permutation matrix such that
\begin{equation}
  \mathbf{A}(k)\mathbf{K}(k)^T=\begin{bmatrix}\hat{\mathbf{A}}(k) &\mathbf{y}(k)\\ \mathbf{z}(k)^T &\beta(k)\\
\end{bmatrix}.\label{eqn:permu}
\end{equation}
In addition, let $\mathbf{\hat{x}}(k)$ denote the solution to the following subsystem:
        \begin{equation}
          \hat{\mathbf{A}}(k)\mathbf{\hat{x}}(k)=-\mathbf{y}(k).\label{eqn:comp-hatx}
        \end{equation}
Then, based on $\mathbf{\hat{x}}(k)$, we can compute $\bm{\pi}(k)$ by the partition factorization method\cite{computepi} in Algorithm~\ref{alg:comp-pi}.

\begin{algorithm}[!h]
\caption{Algorithm to Compute $\bm{\pi}(k)$}
\label{alg:comp-pi}
\begin{algorithmic}[1]
\STATE Obtain $\mathbf{P}(k)$ and $\mathbf{A}(k)$ in \eqref{eqn:Ak}.
\STATE Find $\mathbf{K}(k)$ and partition $\mathbf{A}(k)$ into the form \eqref{eqn:permu} to obtain $\hat{\mathbf{A}}(k)$ and $\mathbf{y}(k)$.
\STATE Compute $\mathbf{\hat{x}}(k)$ using Gaussian elimination.        \label{code:fram:comp-hatxk}
\STATE Let $
  \mathbf{x}(k)\triangleq \mathbf{K}(k)\begin{bmatrix}\mathbf{\hat{x}}(k)\\ 1\\\end{bmatrix}
$ and  normalize $\mathbf{x}(k)$ to obtain $\bm{\pi}(k)$, i.e.,
$
          \bm{\pi}(k)=\frac{\mathbf{x}(k)}{||\mathbf{x}(k)||}.
$
\end{algorithmic}
\end{algorithm}

\begin{remark}[Computational Complexity of Gaussian elimination]

The computation of each $\mathbf{\hat{x}}(k)$ using Gaussian elimination in step \ref{code:fram:comp-hatxk} of Algorithm~\ref{alg:comp-pi} requires $2(|\mathcal{C}|-1)^3/3$ flops.\footnote{The computational complexity is measured as the number of floating-point operations (flops), where a \emph{flop} is defined as one addition, subtraction, multiplication or division of two floating-point numbers\cite{matrix}.} Thus, the computation of  $\{\mathbf{\hat{x}}(k):k=1,2,\cdots,|\mathcal{C}|\}$ using Gaussian elimination requires $2|\mathcal{C}|(|\mathcal{C}|-1)^3/3$ flops, i.e., is of complexity $O(|\mathcal{C}|^4)$.
\label{remark:remarkofComxBrute}
\end{remark}
\begin{figure}[h]
\begin{centering}
\includegraphics[scale=.35]{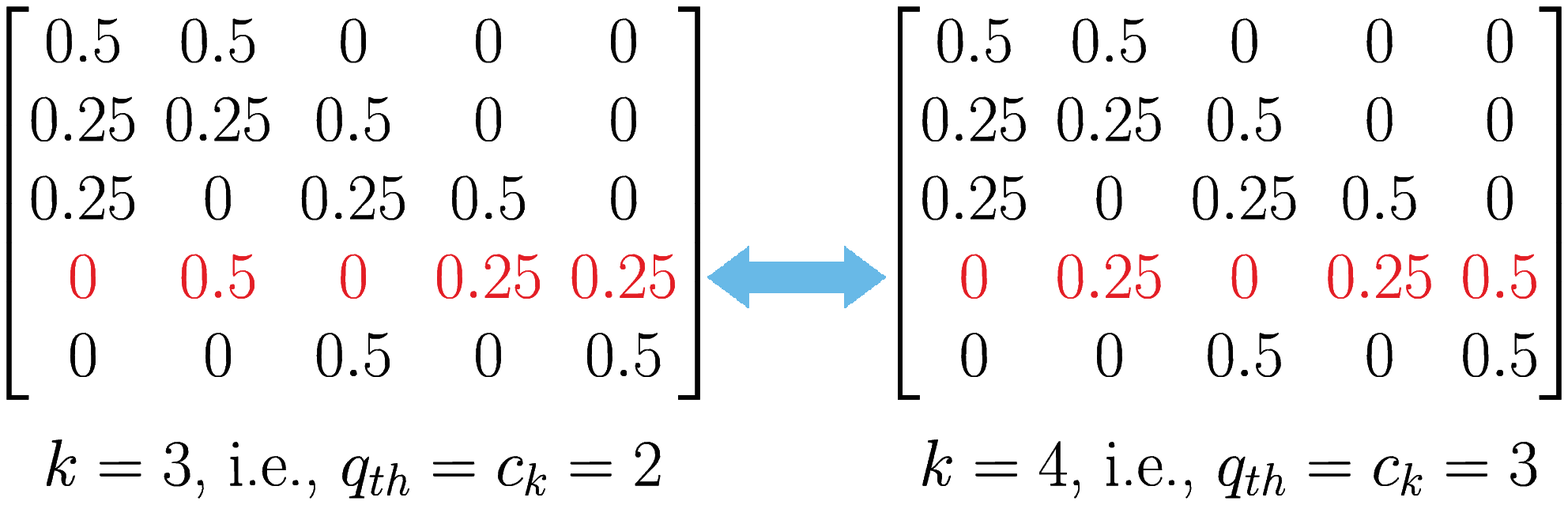}
 \caption{\small{Illustration of $\mathbf{P}(k)$. $R_s=1$, $R_r=2$, $N_r=4$, $p_s=p_r=0.5$. $\mathcal{C}$=$\{0,1,2,3,4\}$.}}\label{fig:illustration}
\end{centering}
\end{figure}

On the other hand, for each $k=1,2,\cdots,|\mathcal{C}|$, $\mathbf{\hat{x}}(k)$ can also be obtained by multiplying both sides of \eqref{eqn:comp-hatx} with $\hat{\mathbf{A}}(k)^{-1}$.\footnote{$\hat{\mathbf{A}}(k)^{-1}$ exists because $\hat{\mathbf{A}}(k)$ is a nonsingular matrix\cite{computepi}.}
This involves matrix inversion. To reduce the complexity, instead of computing $\hat{\mathbf{A}}(k)^{-1}$ for each $k$ separately, we shall compute $\hat{\mathbf{A}}(k)^{-1}$ iteratively (i.e., compute $\hat{\mathbf{A}}(k+1)^{-1}$ based on $\hat{\mathbf{A}}(k)^{-1}$) by exploiting the relationship between $\mathbf{P}(k)$ and $\mathbf{P}(k+1)$.
Specifically, for two adjacent thresholds $c_k$ and $c_{k+1}$, the corresponding transition probability matrices $\mathbf{P}(k)$ and $\mathbf{P}(k+1)$ differ only in the $(k+1)$-th row, as illustrated in Fig.~\ref{fig:illustration}.
The following lemma summarizes the relationship between  $\hat{\mathbf{A}}(k+1)^{-1}$ and $\hat{\mathbf{A}}(k)^{-1}$, which directly results from the special structure of $\mathbf{P}(k)$.
\begin{lemma}[Relationship between $\hat{\mathbf{A}}(k+1)^{-1}$ and $\hat{\mathbf{A}}(k)^{-1}$]
 Let $\hat{\mathbf{K}}(k)$ denote the $(|\mathcal{C}|-1)\times(|\mathcal{C}|-1)$ permutation matrix obtained by exchanging the $(k+1)$-th and $(k+2)$-th columns of $\mathbf{I}_{|\mathcal{C}|-1}$
and let $\mathbf{a}_{k+1}(k+1)$ and $\mathbf{a}_{k+2}(k)$ denote the $(k+1)$-column of $\hat{\mathbf{A}}(k+1)$ and the $(k+2)$-column of $\hat{\mathbf{A}}(k)$, respectively.
Then, $\hat{\mathbf{A}}(k+1)^{-1}$ and $\hat{\mathbf{A}}(k)^{-1}$ satisfy:
\begin{align}
  \hat{\mathbf{A}}(k+1)^{-1}&=\hat{\mathbf{K}}(k)\hat{\mathbf{A}}(k)^{-1}\nonumber\\
  &-\frac{\hat{\mathbf{K}}(k)\hat{\mathbf{A}}(k)^{-1}\mathbf{u}(k)\mathbf{v}(k)^T\hat{\mathbf{K}}(k)\hat{\mathbf{A}}(k)^{-1}}{1+\mathbf{v}(k)^T\hat{\mathbf{K}}(k)\hat{\mathbf{A}}(k)^{-1}\mathbf{u}(k)},\label{eqn:comp-nextk}
\end{align}

where
\begin{align}
    &\mathbf{u}(k)\triangleq \mathbf{a}_{k+1}(k+1)-\mathbf{a}_{k+2}(k),\label{eqn:defiuk}\\
    &\mathbf{v}(k)\triangleq \mathbf{e}_{k+1,|\mathcal{C}|-1}.\label{eqn:defivk}
\end{align}
\label{lemma:lemmaofupdate}
\end{lemma}
\begin{proof}
Please see Appendix F.
\end{proof}

Based on Lemma~\ref{lemma:lemmaofupdate}, we can compute $\mathbf{\hat{x}}(k)$ by Algorithm~\ref{alg:comp-hatx}.
\begin{algorithm}[!h]
\caption{Algorithm to Compute $\mathbf{\hat{x}}(k)$}
\label{alg:comp-hatx}
\begin{algorithmic}[1]
\IF{$k=1$}
\STATE Compute $\hat{\mathbf{A}}(1)^{-1}$ using Gaussian elimination.
\ELSE
\STATE Obtain $\hat{\mathbf{K}}(k-1)$, $\mathbf{u}(k-1)$ and $\mathbf{v}(k-1)$ in Lemma ~\ref{lemma:lemmaofupdate}.
\label{code:fram:const-kuv}
\STATE Compute $\hat{\mathbf{A}}(k)^{-1}$ based on $\hat{\mathbf{A}}(k-1)^{-1}$ according to \eqref{eqn:comp-nextk}.
\label{code:fram:comp-nextk}
\ENDIF
\STATE Compute $\mathbf{\hat{x}}(k)=-\hat{\mathbf{A}}(k)^{-1}\mathbf{y}(k)$.
\label{code:fram:comp-hatxk-2}
\end{algorithmic}
\end{algorithm}
\begin{remark}[Computational Complexity of Algorithm~\ref{alg:comp-hatx}]
  By Algorithm~\ref{alg:comp-hatx}, for $k=1$, the computation of $\mathbf{\hat{x}}(k)$ requires $8(|\mathcal{C}|-1)^3/3+2(|\mathcal{C}|-1)^2$ flops. For each $k=2,3,\cdots,|\mathcal{C}|$, steps \ref{code:fram:const-kuv}, \ref{code:fram:comp-nextk} and \ref{code:fram:comp-hatxk-2} require  $|\mathcal{C}|-1$, $10(|\mathcal{C}|-1)^2$ and $2(|\mathcal{C}|-1)^2$ flops, respectively, and hence the computation of $\mathbf{\hat{x}}(k)$ requires $12(|\mathcal{C}|-1)^2+|\mathcal{C}|-1$ flops.
Therefore, the computation of  $\{\mathbf{\hat{x}}(k):k=1,2,\cdots,|\mathcal{C}|\}$ using Algorithm~\ref{alg:comp-hatx} requires $44(|\mathcal{C}|-1)^3/3+3(|\mathcal{C}|-1)^2$ flops, i.e., is of complexity $O(|\mathcal{C}|^3)$.
\label{remark:remarkofComxUpdate}
\end{remark}

By comparing Remarks~\ref{remark:remarkofComxBrute} and~\ref{remark:remarkofComxUpdate}, we can see that, the complexity of computing $\{\mathbf{\hat{x}}(k):k=1,2,\cdots,|\mathcal{C}|\}$ using Algorithm~\ref{alg:comp-hatx} $\left(O(|\mathcal{C}|^3)\right)$ is lower than that using Gaussian elimination in step \ref{code:fram:comp-hatxk} of Algorithm~\ref{alg:comp-pi} $\left(O(|\mathcal{C}|^4)\right)$. This is because using Gaussian elimination in step \ref{code:fram:comp-hatxk} of Algorithm~\ref{alg:comp-pi} cannot make use of the special structure of $\mathbf{P}(k)$, and hence has higher computational complexity.

By replacing step \ref{code:fram:comp-hatxk} in Algorithm~\ref{alg:comp-pi} with Algorithm~\ref{alg:comp-hatx}, we can compute $\bm{\pi}(k)$ for all $k$ iteratively. Therefore, we can develop Algorithm~\ref{alg:total} to solve Problem~\ref{problem:equivalentproblem}.
\begin{algorithm}[h]
\caption{Algorithm to Compute $q_{th}^*$ for Problem~\ref{problem:equivalentproblem}}
\label{alg:total}
\begin{algorithmic}[1]
\STATE \textbf{initialize} $q^*_{th}=0$, $temp=0$.
\FOR{$k=1:|\mathcal{C}|$}
\STATE $q_{th}\leftarrow c_k$.
\STATE Compute $\mathbf{r}(q_{th})$ by \eqref{eqn:comp-ri}.
\STATE Compute $\bm{\pi}(q_{th})$ by Algorithm~\ref{alg:comp-pi} wherein step \eqref{code:fram:comp-hatxk} is replaced with Algorithm~\ref{alg:comp-hatx}.
\label{code:fram:comp-pi-qth}
\IF{$\bm{\pi}(q_{th})\mathbf{r}(q_{th})\geq temp$} \STATE $temp\leftarrow\bm{\pi}(q_{th})\mathbf{r}(q_{th})$, $q^*_{th}\leftarrow q_{th}$. \ENDIF
\ENDFOR
\end{algorithmic}
\end{algorithm}

\section{Optimal Solution for Special Case}
In this section, we first obtain the corresponding static optimization problem for the symmetric case ($R_s=R_r=R, N_r=nR$ and $p_s=p_r=p$). Then, we derive its closed-form optimal solution.

By Lemma 2, the recurrent class of $\{Q_t\}$ is given by $\mathcal{C}=\{0, R, 2R,\cdots,nR\}$.  Fig.~\ref{fig:specialcase} illustrates the corresponding transition diagram. By applying the Perron-Frobenius theorem and the detailed balance equations\cite{gallager}, we obtain the steady-state probability:
\begin{subequations}
\begin{align}
&\pi_{0}(m)=\frac{\bar{p}^{2m+2}-\bar{p}^{2m+1}}{\bar{p}^{2m+2}+\bar{p}^{n+1}-2\bar{p}^{m+1}},\\
&\pi_{i+1}(m)=
\begin{cases}\frac{1}{\bar{p}}\pi_{i}(m),&0\leq i\leq m-1;\\
\pi_{i}(m),&i=m;\\
\bar{p}\pi_{i}(m),&m+1\leq i\leq n-1.
\end{cases}
\end{align}\label{eqn:pi0pin}
\end{subequations}
where $\bar p\triangleq 1-p$.
Then, in the symmetric case, Problem~\ref{problem:equivalentproblem} is equivalent to the following optimization problem.
\begin{problem}[Optimization   for Symmetric Case]
  \begin{align}
\min_{m\in\{0,1,...,n\}} &~~~\frac{\bar{p}^{n+1}-\bar{p}^n+\bar{p}^{2m+2}-\bar{p}^{2m+1}}{\bar{p}^{2m+2}+\bar{p}^{n+1}-2\bar{p}^{m+1}}.\label{eqn:optimalm}
\end{align}
\label{problem:problemSymm}
\end{problem}
\begin{figure}[t]
\begin{centering}
\includegraphics[scale=.41]{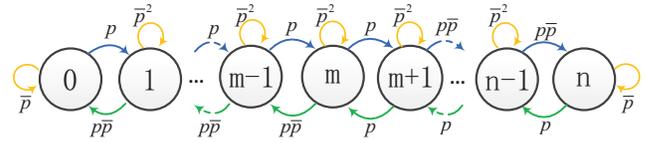}
 \caption{\small{The transition diagram of $\{Q_t\}$ for the symmetric case. State $i\in\{0,1,\cdots,n\}$ represents state $iR\in\mathcal{C}$ and $m=q_{th}/R$.}}\label{fig:specialcase}
\end{centering}
\end{figure}
By change of variables, we can equivalently transform the discrete optimization problem in  Problem~\ref{problem:problemSymm} to a continuous optimization problem and obtain the optimal threshold to Problem~\ref{problem:originalproblem}, which is  summarized in the following lemma.
\begin{lemma}[Optimal Threshold for Symmetric Case]
  In the symmetric case, any threshold
\begin{eqnarray}
  Q_{th}^*\in\begin{cases}\{\frac{n-1}{2}R,\frac{n-1}{2}R+1,\cdots,\frac{n+1}{2}R-1\},&\textrm n~\text{is odd}\\
\{\frac{nR}{2}-R,\frac{nR}{2}-R+1,\cdots,\frac{nR}{2}+R-1\},&\textrm n~\text{is even}
\end{cases}\nonumber
\end{eqnarray} achieves the optimal value to Problem~\ref{problem:originalproblem}.
\label{lemma:solution2Symm}
\end{lemma}

\begin{proof}
  Please see Appendix G.
\end{proof}
\section{Numerical Results and Discussion}
In this section, we verify  the analytical results and evaluate the performance of the proposed optimal solution via numerical examples.  In the simulations, we choose $p_s=p_r=0.5$.

\subsection{Threshold Structure of Optimal Policy}
Fig.~\ref{fig:1:V} illustrates the value function  $V(Q)$ versus $Q$.
$V(Q)$ is  computed numerically using RVIA \cite{bertsekas}.
 It can be seen that $V(Q)$ is increasing with $Q$ and $V(Q+1)-V(Q)\leq 1$, which verify Properties 1) and 2) in Lemma~\ref{lemma:propertiesofV}, respectively.
The third property of Lemma~\ref{lemma:propertiesofV} can also be verified by checking the simulation points.
Fig.~\ref{fig:1:deltaJ} illustrates the function $\Delta J(Q)\triangleq J(Q,1)-J(Q,0)$ versus $Q$. Note that, $\Delta J(Q)$ is a function of $V(Q)$, which is computed numerically using RVIA (a standard numerical MDP technique). According to the Bellman equation in Theorem~\ref{theorem:theorem2}, $\Delta J(Q)\geq 0$ indicates that it is optimal to schedule the R-D link for state $Q$; $\Delta J(Q)<0$ indicates that it is optimal to schedule the S-R link for state $Q$. Hence, from Fig.~\ref{fig:1:deltaJ}, we know that the optimal policy (obtained using RVIA) has a threshold-based structure and $Q^*_{th}=3, 7, 12$ are the optimal thresholds for the three cases. We have also calculated the optimal threshold for the three cases, using Algorithm~\ref{alg:total} for the general case ($R_s=1, R_r=2$ and $R_s=2, R_r=1$) and Lemma~\ref{lemma:solution2Symm} for the symmetric case ($R_s=1, R_r=1$). The obtained thresholds are equal to the optimal values obtained by the numerical MDP technique.
\begin{figure}[h]
\begin{minipage}[t]{.5\linewidth}
\centering
        \includegraphics[height=4.1cm, width=4.4cm]{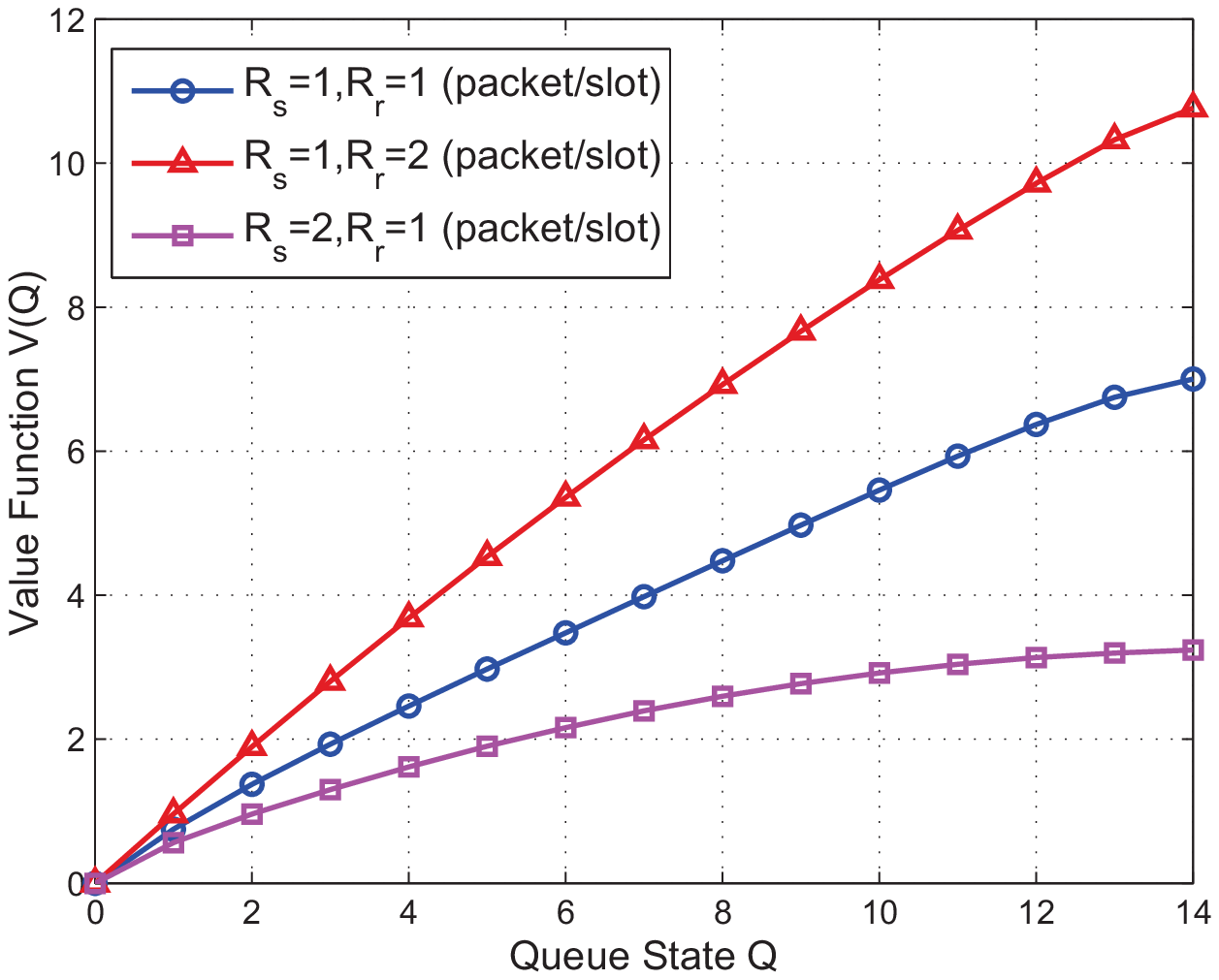}
\subcaption{\small{$V(Q)$}\label{fig:1:V}}
\end{minipage}%
\begin{minipage}[t]{.5\linewidth}
\centering
\includegraphics[height=4.1cm, width=4.4cm]{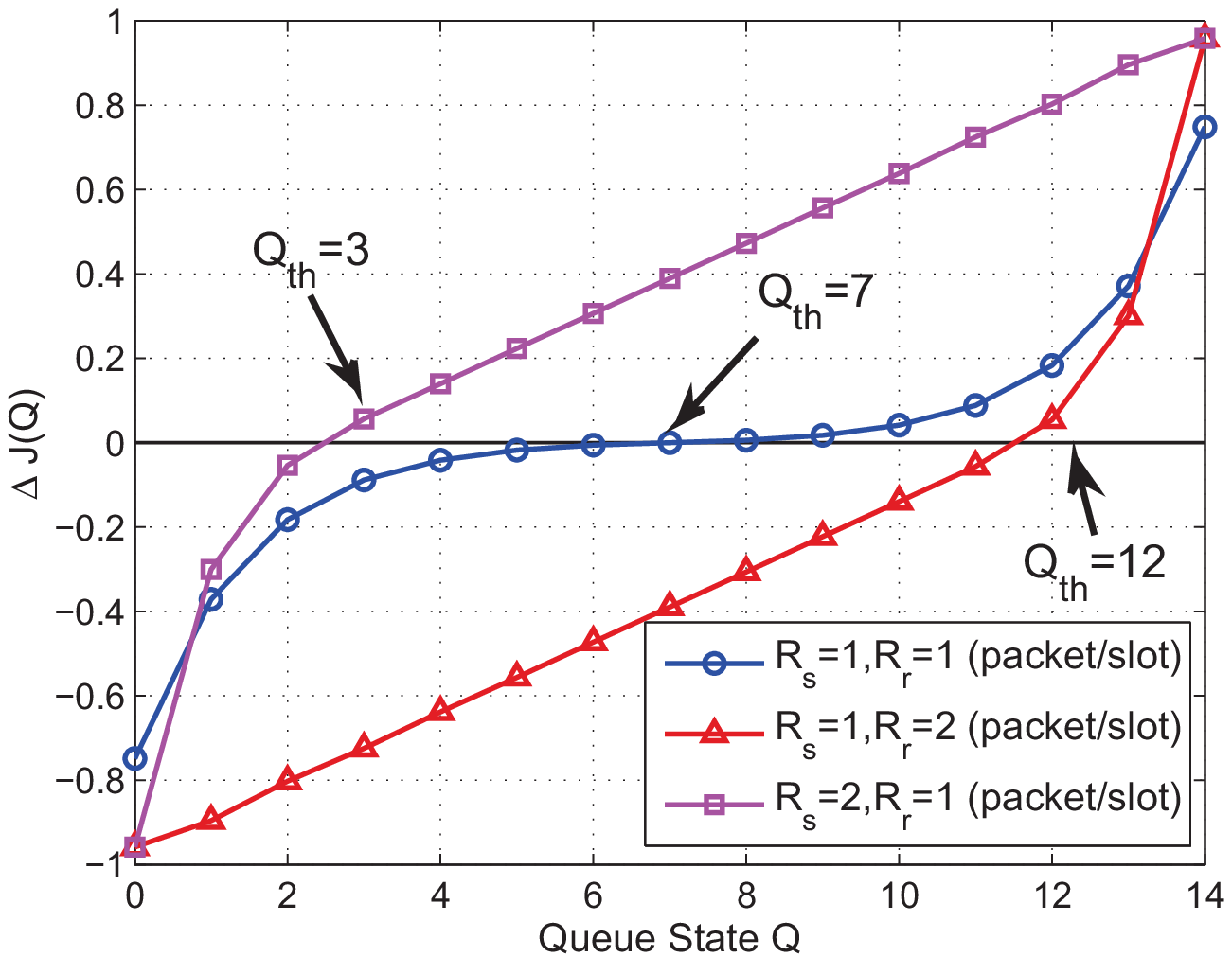}
\subcaption{\small{$\Delta J(Q)$}}\label{fig:1:deltaJ}
\end{minipage}
\caption{\small{Verification of analytical results. $N_r$$=$$14$ packets.}}\label{fig:1}
\end{figure}

\subsection{Throughput Performance}
We compare the throughput performance of the proposed optimal policy (given in Theorems \ref{theorem:theorem1} and \ref{theorem:theorem2})  with five baseline schemes: DOPN, ADOP, TOP, OLSP and NOP.\footnote{The detailed illustrations of DOPN, ADOP, TOP and OLSP are given in Section I.} In particular, DOPN refers to the Delay-Optimal Policy for Non-fading channels in \cite{ISIT11Cui}, and ADOP refers to the Asymptotically Delay-Optimal Policy for on/off fading channels in \cite{TIT15Cui}, both of which are designed for two-hop networks with infinite buffers at the source and relay. TOP refers to the Throughput-Optimal Policy for a multi-hop network with infinite source buffers and finite relay buffers in \cite{Long}. OLSP refers to the Optimal Link Selection Policy for a two-hop system with an infinite relay buffer in \cite[Theorem 2]{ztit}. NOP refers to the Near-Optimal Policy obtained based on approximate value iteration using \emph{aggregation}\cite[Chapter 6.3]{bertsekas}, which is similar to the approximate MDP technique used in \cite{wangtsp} and\cite{wangtit}. Note that, OLSP depends on the CSI only, while the other four baseline schemes depend on both of the CSI and QSI.
In addition, the threshold in DOPN ($Q_{th}$$=$$0$) is fixed; the threshold in ADOP ($Q_{th}$$=$$R_r$) depends on $R_r$; the threshold in TOP ($Q_{th}$$=$$N_r/2$) depends on $N_r$; NOP adapts to $R_s,R_r$ and $N_r$.

Fig.~\ref{fig:fixedbuffer2} and Fig.~\ref{fig:fixedrate2} illustrate the average system throughput versus the maximum transmission rate and  the relay buffer size, respectively, in the asymmetric case ($R_s$$\neq$$R_r$).
Since DOPN, ADOP, TOP, NOP and the proposed optimal policy depend on both of the CSI and QSI, they can achieve better throughput performance than OLSP in most cases.
Moreover, as the threshold in the proposed optimal policy also depends on $R_s,R_r$ and $N_r$, it outperforms all the baseline schemes.
In summary, the proposed optimal policy can make better use of the system information and system parameters, and hence achieves the optimal throughput.
Specifically, the performance gains of the proposed policy over DOPN, ADOP, TOP, OLSP and NOP are up to $15\%$, $10\%$, $80\%$, $17\%$ and $8\%$, respectively.
Besides, the performance of TOP relies heavily on the choice for the parameter $R_{max}$ (the maximum admitted rate), which is not specified in \cite{Long}.
\begin{figure}[h]
\begin{minipage}[t]{.5\linewidth}
\centering
        \includegraphics[height=4.1cm, width=4.4cm]{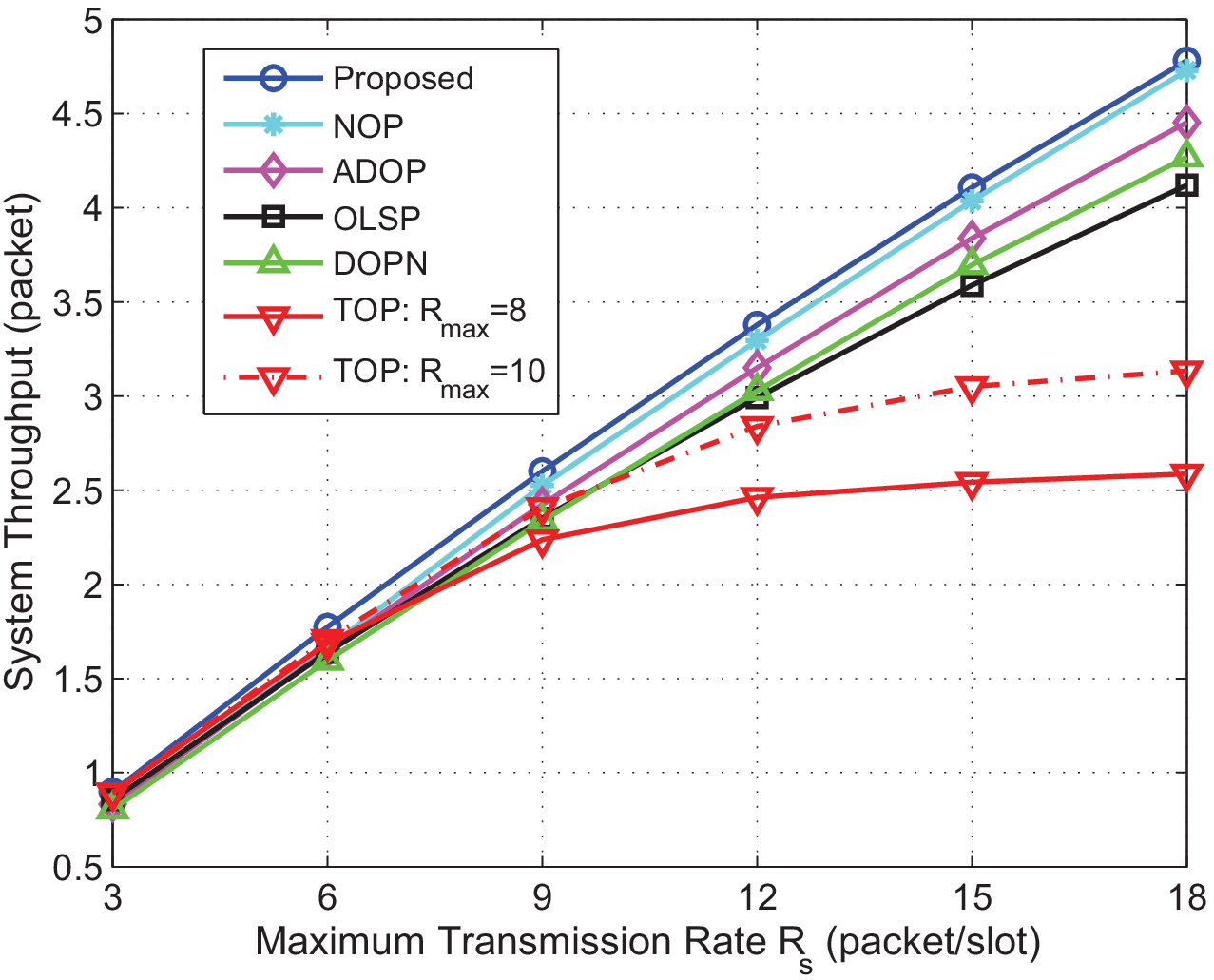}
\subcaption{\small{Throughput versus $R_s$.$~~~~~~~$ $~~~~~N_r=50$ packets.}}\label{fig:fixedbuffer2}
\end{minipage}%
\begin{minipage}[t]{.5\linewidth}
\centering
\includegraphics[height=4.1cm, width=4.4cm]{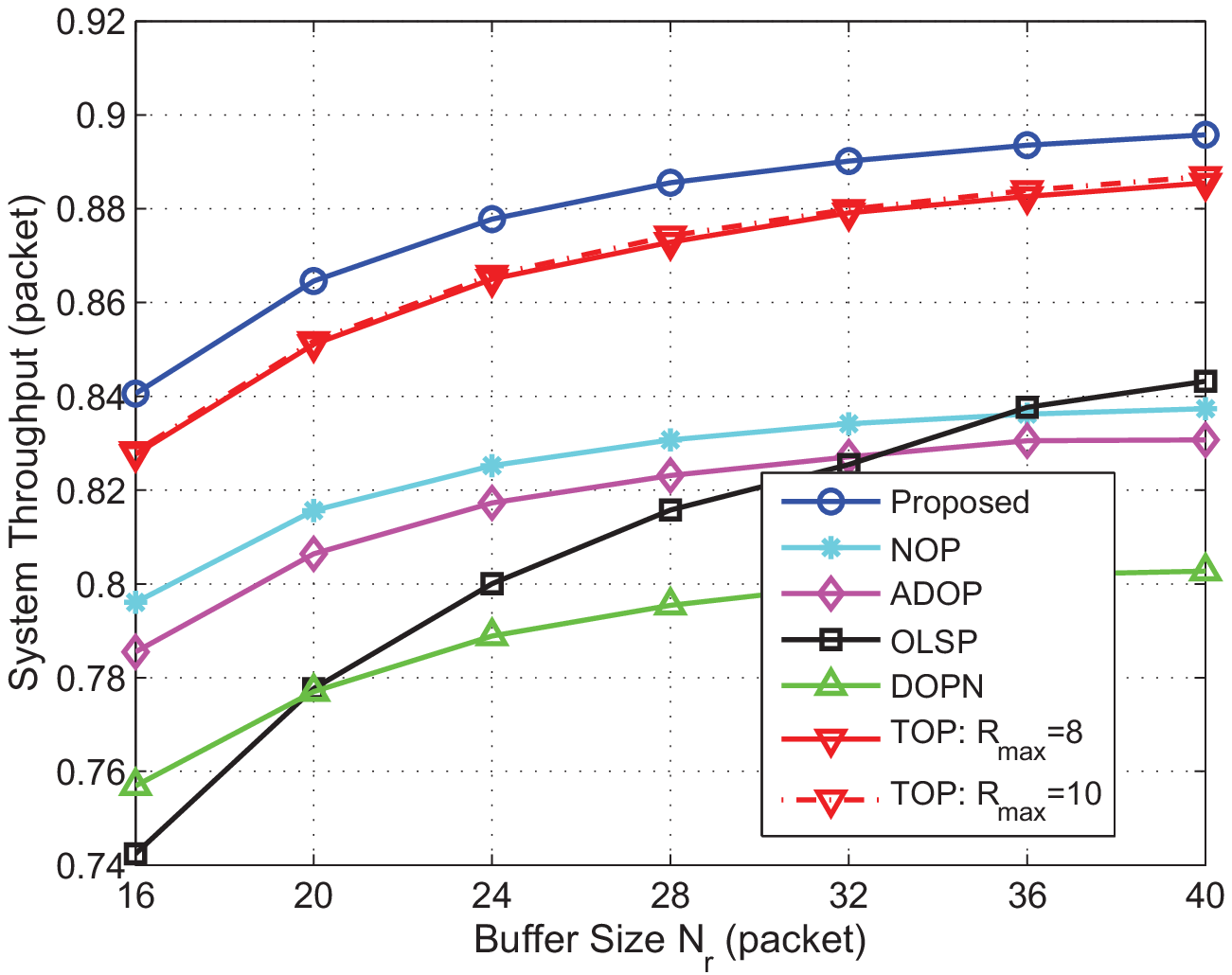}
\subcaption{\small{Throughput versus $N_r$.$~~~~~~~$ $~~~~~R_s=3$ packets/slot.}}\label{fig:fixedrate2}
\end{minipage}
\caption{\small{Throughput for different schemes in the asymmetric case ($R_s$$\neq $$R_r$). $R_s/R_r$$=$$3/2$. The unit of $R_{max}$ is packet/slot.}}\label{fig:fig2}
\end{figure}

Fig.~\ref{fig:fixedbuffer} and Fig.~\ref{fig:fixedrate}  illustrate the average system throughput versus the maximum transmission rate and  the relay buffer size, respectively, in the symmetric case ($R_s$$=$$R_r$). Similar observations can be made for the symmetric case. The proposed optimal policy outperforms all the baseline schemes and its performance gains over DOPN, ADOP, TOP, OLSP and NOP are up to $20\%$, $15\%$, $30\%$, $20\%$ and $13\%$, respectively.

\begin{figure}[t]
\begin{minipage}[t]{.5\linewidth}
\centering
        \includegraphics[height=4.1cm, width=4.4cm]{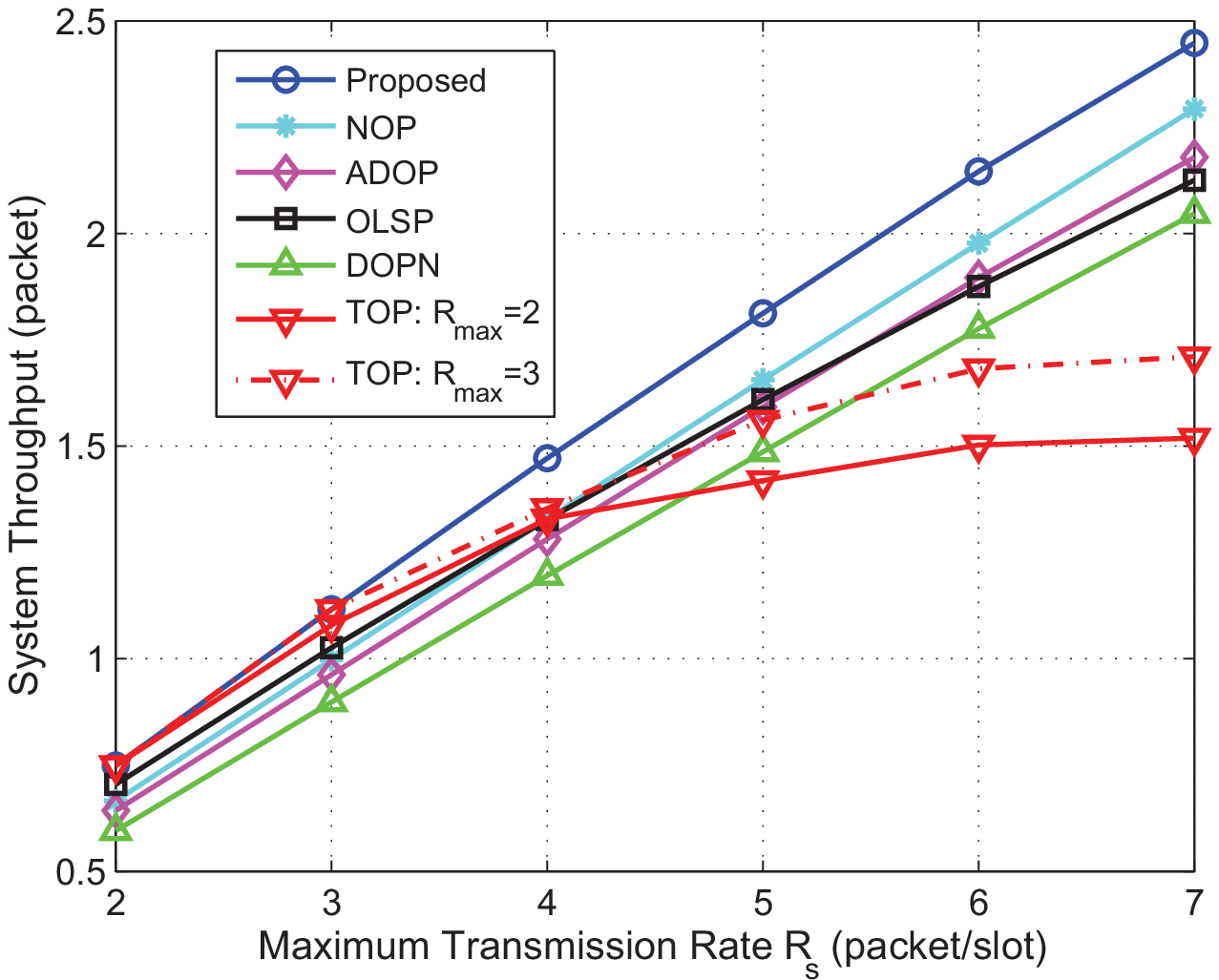}
        \centering
\subcaption{\small{Throughput versus $R_s$.$~~~~~~~~$ $~~~~~N_r=30$ packets.}}\label{fig:fixedbuffer}
\end{minipage}%
\begin{minipage}[t]{.5\linewidth}
\centering
\includegraphics[height=4.1cm, width=4.4cm]{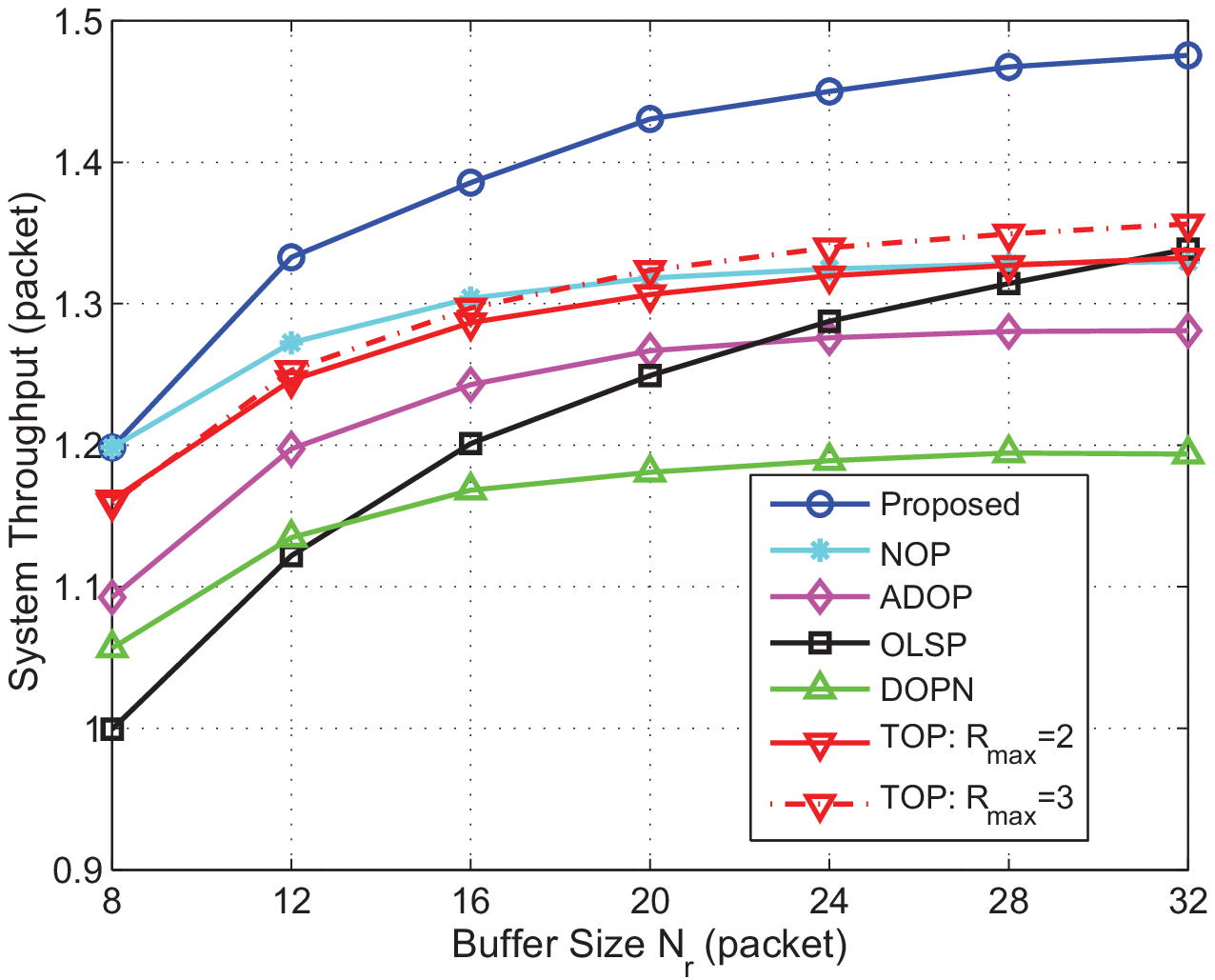}
\subcaption{\small{Throughput versus $N_r$.$~~~~~~~~$ $~~~~~R_s=4$ packets/slot.}}\label{fig:fixedrate}
\end{minipage}
\caption{\small{Throughput performance for different schemes in the symmetric case ($R_s$$=$$R_r$). The unit of $R_{max}$ is packet/slot.}}\label{fig:fig3}
\end{figure}

\subsection{Computational Complexity}

Table~\ref{tableasym} illustrates the average Matlab computation time of different algorithms in the asymmetric case ($R_s$$\neq$$R_r$).
It can be seen that, our proposed Algorithm~\ref{alg:total} achieves the lowest computational complexity.
Specifically, the standard numerical algorithms (i.e., policy iteration and relative value iteration) designed for the stochastic optimization problem (Problem~\ref{problem:originalproblem}) have much higher computational complexity than the algorithms (i.e., the brute-force algorithm and Algorithm~\ref{alg:total}) for the static optimization problem (Problem~\ref{problem:equivalentproblem}). In addition, for Problem~\ref{problem:equivalentproblem}, the complexity of the brute-force algorithm is higher than that of the proposed Algorithm~\ref{alg:total} and the complexity gap between them increases with $|\mathcal{C}|$ rapidly. This verifies the discussions in Remarks~\ref{remark:remarkofComxBrute} and \ref{remark:remarkofComxUpdate}.

\begin{table}[!htbp]
\footnotesize
\begin{adjustbox}{max width=0.49\textwidth}
\begin{tabular}{cccccc}
\toprule
\multirow{2}{*}{\tabincell{c}{Buffer Size\\ (packet)}}& \multicolumn{2}{c}{Stochastic Opt. (Problem~\ref{problem:originalproblem})} &&\multicolumn{2}{c}{Static Opt¡£ (Problem~\ref{problem:equivalentproblem})} \\\cmidrule{2-3}\cmidrule{5-6}
& PIA&RVIA   &&Brute-Force Alg.  &Alg. ~\ref{alg:total}\\ \midrule
\tabincell{c}{$N_r=40$\\ $(|\mathcal{C}|=21)$} & 0.4005&0.1754&& 0.0166 & 0.0154 \\
\midrule
\tabincell{c}{$N_r=60$\\ $(|\mathcal{C}|=31)$} & 0.7785&0.4336&& 0.0470 & 0.0443 \\
\midrule
\tabincell{c}{$N_r=80$\\ $(|\mathcal{C}|=41)$} & 0.8643&0.4993&& 0.0734 & 0.0635 \\
\midrule
\tabincell{c}{$N_r=100$\\ $(|\mathcal{C}|=51)$} & 1.0698&0.6253&& 0.1437 & 0.0988 \\
\bottomrule
\end{tabular}
\end{adjustbox}
\centering
\caption{\small{Average Matlab computation time (sec) comparison in the asymmetric case. $R_s=4$ packets/slot and $R_r=2$ packets/slot. Policy iteration algorithm (PIA) and RVIA are two standard numerical algorithms to solve the stochastic optimization problem (Problem~\ref{problem:originalproblem}) based on the equivalent Bellman equation in \eqref{eqn:reduceAction}. The brute-force algorithm and Algorithm~\ref{alg:total} are designed to solve the static optimization problem (Problem~\ref{problem:equivalentproblem}), as illustrated in Section V-B.}}
\label{tableasym}
\end{table}

Table~\ref{tablesym} illustrates the average Matlab computation time of different algorithms in the symmetric case ($R_s$$=$$R_r$).
It can be seen that, the numerical algorithms for Problem~\ref{problem:originalproblem} have much higher computational complexity than the proposed solution for Problem~\ref{problem:equivalentproblem}.
Note that, in the symmetric case, Problem~\ref{problem:equivalentproblem} has a closed-form solution, as shown in Lemma \ref{lemma:solution2Symm}. Thus, the computation time of the closed-form solution is negligible and does not change with $|\mathcal{C}|$.

\begin{table}[!htbp]
\footnotesize
\begin{adjustbox}{max width=0.49\textwidth}
\begin{tabular}{ccccc}
\toprule
 \multirow{2}{*}{\tabincell{c}{Buffer Size\\ (packet)}}& \multicolumn{2}{c}{Stochastic Opt. (Problem~\ref{problem:originalproblem})} &&Static Opt. (Problem~\ref{problem:equivalentproblem}) \\\cmidrule{2-3}\cmidrule{5-5}
&PIA & RVIA  && Closed-form in Lemma~\ref{lemma:solution2Symm}  \\ \midrule
\tabincell{c}{$N_r=40$\\ $(|\mathcal{C}|=21)$} &0.2560&0.1799&& 0.000001 \\
\midrule
\tabincell{c}{$N_r=60$\\ $(|\mathcal{C}|=31)$} &0.4951&0.2421&& 0.000001 \\
\midrule
\tabincell{c}{$N_r=80$\\ $(|\mathcal{C}|=41)$} &0.9612&0.4478&& 0.000001 \\
\midrule
\tabincell{c}{$N_r=100$\\ $(|\mathcal{C}|=51)$} &1.7854&0.8712&& 0.000001 \\
\bottomrule
\end{tabular}
\end{adjustbox}
\centering
\caption{\small{Average Matlab computation time (sec) comparison in the symmetric case. $R_s=R_r=2$ packets/slot.  The static optimization problem (Problem~\ref{problem:equivalentproblem}) has a closed-form solution, as shown in Lemma~\ref{lemma:solution2Symm}}}\label{tablesym}
\end{table}

\section{Conclusion}
In this paper, we consider the optimal  control to maximize the  average system throughput for a two-hop half-duplex relaying system with random channel connectivity and a finite relay buffer. We formulate the stochastic optimization problem as an infinite horizon average cost MDP. Then, we show that the optimal link selection policy has a threshold-based structure. Based on the structural properties of the optimal policy,  we simplify the MDP to a static discrete optimization problem and propose a low-complexity algorithm to obtain the optimal threshold. Furthermore,  we  obtain the closed-form optimal threshold for the symmetric case.

\appendices
\section*{Appendix A: Proof of Theorem~\ref{theorem:theorem1}}\label{app:theorem1}
First, using sample path arguments, we show that the link selection and transmission rate control policy $\Omega^*=(\alpha^*,\mu^*)$ is optimal, where $\alpha^*$ and $\mu^*$ satisfy the structures in \eqref{eqn:reducealpha} and \eqref{eqn:rate}, respectively.

Consider any stationary link selection and transmission rate control policy $\Omega=(\alpha,\mu)$ satisfying Definition~\ref{definition:definition1}.
Let $\{G_t\}$ be a given CSI sample path.
Denote $(a_{s,t},a_{r,t})$ and $(u_{s,t},u_{r,t})$  be the link selection and transmission rate control action at slot $t$ under $\Omega$, respectively.
Let $\{Q_t\}$ be the associated trajectory of QSI which evolves according to \eqref{eqn:queue-dyn-t} with $\{(a_{s,t},a_{r,t})\}$ and $\{(u_{s,t},u_{r,t})\}$.
Denote $(a_{s,t}^*,a_{r,t}^*)$ and $(u_{s,t}^*,u_{r,t}^*)$ be another link selection and transmission rate control action at slot $t$, respectively.
Let $\{Q_t^*\}$ be the associated trajectory of QSI which evolves according to \eqref{eqn:queue-dyn-t} with $\{(a_{s,t}^*,a_{r,t}^*)\}$ and $\{(u_{s,t}^*,u_{r,t}^*)\}$.
Assume $Q_1^*=Q_1$.
The relationship between $(a_{s,t}^*,a_{r,t}^*)$ and $(a_{s,t},a_{r,t})$ is given by:
\begin{equation}\label{eqn:relationship}
  (a_{s,t}^*,a_{r,t}^*)=\left\{
    \begin{array}{ll}
      (0,1)~\text{or}~(1,0), & \hbox{\parbox[t]{.15\textwidth}{if $\mathbf{G}_t$$=$$(1,1)~\text{and}$ $~~~~~~(a_{s,t},a_{r,t})$$=$$(0,0)$}} \\
      (a_{s,t},a_{r,t}), & \hbox{otherwise}
    \end{array}
  \right.
\end{equation}
$(u_{s,t}^*,u_{r,t}^*)$ satisfies the structure in \eqref{eqn:rate}, i.e., $u_{s,t}^*=\min\{R_s, N_r-Q_t^*\}$ and $u_{r,t}^*=\min\{R_r, Q_t^*\}$.
We shall show that the throughput under $\{(a_{s,t}^*,a_{r,t}^*)\}$ and $\{(u_{s,t}^*,u_{r,t}^*)\}$  is no smaller than that under $\{(a_{s,t},a_{r,t})\}$ and $\{(u_{s,t},u_{r,t})\}$ for a given CSI sample path $\{G_t\}$.
Define $\Delta_t\triangleq\sum_{\tau=1}^t\left(a_{r,t}^*u_{r,t}^*-a_{r,t}u_{r,t}\right)$.
It is equivalent to prove $\Delta_t\geq 0$ for all $t$.
In the following, using mathematical induction, we shall show that $\Delta_t\geq0$ and $\Delta_t+Q_{t+1}^*\geq Q_{t+1}$ hold for all $t$. (Note that $\Delta_t+Q_{t+1}^*\geq Q_{t+1}$ is needed to prove $\Delta_t\geq0$.)

Consider $t=1$.  We have $\Delta_1=a_{r,1}^*u_{r,1}^*-a_{r,1}u_{r,1}$ and $\Delta_1+Q_2^*- Q_2=a_{r,1}^*u_{r,1}^*-a_{r,1}u_{r,1}+Q_1^*+a_{s,1}^*u_{s,1}^*-a_{r,1}^*u_{r,1}^*-(Q_1+a_{s,1}u_{s,1}-a_{r,1}u_{r,1})=$$a_{s,1}^*u_{s,1}^*$$-$$a_{s,1}u_{s,1}$.
To prove $\Delta_1\geq 0$ and $\Delta_1+Q_2^*\geq Q_2$, we consider the following two cases.

(1) Consider $(a_{s,1}^*,a_{r,1}^*)=(a_{s,1},a_{r,1})$. Since $u_{r,1}^*=\min\{R_r,Q_1^*\}$, $u_{r,1}\in\{0,1,\cdots,\min\{R_r,Q_1\}\}$ and $Q_1^*=Q_1$, we have $\Delta_1=a_{r,1}(u_{r,1}^*-u_{r,1})\geq 0$. Since $u_{s,1}^*=\min\{R_s, N_r-Q_1^*\}$, $u_{s,1}\in\{0,1,\cdots,\min\{R_s, N_r-Q_1\}\}$ and $Q_1^*=Q_1$, we have $\Delta_1+Q_2^*- Q_2=a_{s,1}(u_{s,1}^*-u_{s,1})\geq 0$.

(2) Consider $(a_{s,1}^*,a_{r,1}^*)\neq(a_{s,1},a_{r,1})$, which implies $(a_{s,1}^*,a_{r,1}^*)=(0,1)$ or $(1,0)$, $(a_{s,1},a_{r,1})=(0,0)$ and $\mathbf{G}_1=(1,1)$ by \eqref{eqn:relationship}. Thus, we have $\Delta_1=a_{r,1}^*u_{r,1}^*\geq 0$ and $\Delta_1+Q_2^*- Q_2=a_{s,1}^*u_{s,1}^*\geq 0$.

Consider $t>1$. Assume $\Delta_{t-1}\geq0$ and $\Delta_{t-1}+Q_t^*\geq Q_t$ hold for some $t>1$.
Note that,
\begin{align}
  &\Delta_{t}=\Delta_{t-1}+a_{r,t}^*u_{r,t}^*-a_{r,t}u_{r,t},\label{eqn:deltat}\\
  &\Delta_t+Q_{t+1}^*-Q_{t+1}=\left(\Delta_{t-1}+a_{r,t}^*u_{r,t}^*-a_{r,t}u_{r,t}\right)\nonumber\\
  &+\left(Q_t^*+a_{s,t}^*u_{s,t}^*-a_{r,t}^*u_{r,t}^*\right)-\left(Q_t+a_{s,t}u_{s,t}-a_{r,t}u_{r,t}\right)\nonumber\\
  &=\Delta_{t-1}+Q_t^*+a_{s,t}^*u_{s,t}^*-Q_t-a_{s,t}u_{s,t}.\label{eqn:deltat2}
\end{align}

To show that $\Delta_{t}\geq0$ and $\Delta_{t}+Q_{t+1}^*\geq Q_{t+1}$ also hold, we consider the following two cases.

(1) If $(a_{s,t}^*,a_{r,t}^*)=(a_{s,t},a_{r,t})$, we consider three cases. (i) If $(a_{s,t},a_{r,t})=(0,0)$, we have $\Delta_{t}=\Delta_{t-1}\geq 0$ and $\Delta_t+Q_{t+1}^*-Q_{t+1}=\Delta_{t-1}+Q_t^*-Q_t\geq 0$. (ii) If $(a_{s,t},a_{r,t})=(1,0)$, we have  $\Delta_{t}=\Delta_{t-1}\geq 0$. Since $u_{s,t}^*=\min\{R_s, N_r-Q_t^*\}$ and $u_{s,t}\in\{0,1,\cdots,\min\{R_s, N_r-Q_t\}\}$, by \eqref{eqn:deltat2}, we have $\Delta_t+Q_{t+1}^*-Q_{t+1}\geq \Delta_{t-1}+Q_t^*+\min\{R_s, N_r-Q_t^*\}-Q_t-\min\{R_s, N_r-Q_t\}=\min\{\Delta_{t-1}+Q_t^*+R_s, \Delta_{t-1}+N_r\}-\min\{Q_t+R_s,N_r\}\geq 0$, where the last inequality is due to the induction hypotheses. (iii) If $(a_{s,t},a_{r,t})=(0,1)$, we have $\Delta_t+Q_{t+1}^*-Q_{t+1}=\Delta_{t-1}+Q_t^*-Q_t\geq 0$. Since $u_{r,t}^*=\min\{R_r, Q_t^*\}$ and $u_{r,t}\in\{0,1,\cdots,\min\{R_r,Q_t\}\}$, by \eqref{eqn:deltat}, we have $\Delta_{t}\geq \Delta_{t-1}+\min\{R_r, Q_t^*\}-\min\{R_r,Q_t\}=\min\{\Delta_{t-1}+R_r, \Delta_{t-1}+Q_t^*\}-\min\{R_r,Q_t\}\geq 0$, where the last inequality is due to the induction hypotheses.

(2) If $(a_{s,t}^*,a_{r,t}^*)\neq(a_{s,t},a_{r,t})$, by \eqref{eqn:relationship}, we have $(a_{s,t}^*,a_{r,t}^*)=(0,1)$ or $(1,0)$, $(a_{s,t},a_{r,t})=(0,0)$ and $\mathbf{G}_t=(1,1)$. By \eqref{eqn:deltat} and \eqref{eqn:deltat2}, we have $\Delta_{t}=\Delta_{t-1}+a_{r,t}^*u_{r,t}^*\geq 0$ and $\Delta_t+Q_{t+1}^*-Q_{t+1}=\Delta_{t-1}+Q_t^*-Q_t+a_{s,t}^*u_{s,t}^*\geq 0$, where the two inequalities are due to the induction hypotheses.

Thus, we show that $\Delta_{t}\geq0$ and $\Delta_{t}+Q_{t+1}^*\geq Q_{t+1}$ also hold.
By induction, $\Delta_t\geq 0$ hold for all $t$ which leads to
\begin{equation}
  \frac{1}{T}\sum_{t=1}^Ta_{r,t}^*u_{r,t}^*\geq \frac{1}{T}\sum_{t=1}^Ta_{r,t}u_{r,t},~\forall T.
\end{equation}
By taking expectation over all sample paths, $\liminf$ and optimization over all link selection and transmission rate control policy space, we have $\max_{\Omega^*}\bar{R}^{\Omega^*} \geq \max_{\Omega}\bar{R}^{\Omega}$, where $\Omega^*=(\alpha^*,\mu^*)$ with $\alpha^*$ and $\mu^*$ satisfying the structures in \eqref{eqn:reducealpha} and \eqref{eqn:rate}, respectively.
In the following, we can restrict our attention to the optimal stationary policy $\Omega^*$.

Problem~\ref{problem:originalproblem} is an infinite horizon average cost MDP. We consider stationary unichain policies.
By Proposition 4.2.5 in \cite{bertsekas}, the Weak Accessibly condition holds for stationary unichain policies. Thus, by Proposition 4.2.3 and Proposition 4.2.1 in \cite{bertsekas}, the optimal average system throughput of the MDP in Problem~\ref{problem:originalproblem} is the same for all initial states and the solution $(\theta,\{\hat{V}(\mathbf{X})\})$ to the following Bellman equation exists.
\begin{align}
  \theta+\hat{V}(\mathbf{X})&=\max_{\Omega(\mathbf{X})}\bigg\{r(\mathbf{X},\Omega(\mathbf{X}))\nonumber\\
 &+\sum_{\mathbf{X}'}\Pr[\mathbf{X}'|\mathbf{X},\Omega(\mathbf{X})]\hat{V}(\mathbf{X}')\bigg\},~\forall \mathbf{X}\in\bm{\mathcal{X}},\label{eqn:oriBellman}
\end{align}
where $\theta=\bar{R}^*$ is the optimal value to Problem~\ref{problem:originalproblem} for all initial state $\mathbf{X}_1\in\bm{\mathcal{X}}$ and $\hat{V}(\cdot)$ is the value function.
Due to the i.i.d. property of $\mathbf{G}$, by taking expectation over $\mathbf{G}$ on both sides of \eqref{eqn:oriBellman}, we have
\begin{align}
  &\theta+V(Q)=\sum_{\mathbf{g}\in\bm{\mathcal{G}}}\Pr(\mathbf{G}=\mathbf{g})\max_{\Omega(\mathbf{X})}\Big\{r(\mathbf{X},\Omega(\mathbf{X}))\nonumber\\
  &~~~~~~~~~~~~+\sum_{Q'}\Pr[Q'|\mathbf{X},\Omega(\mathbf{X})]V(Q')\Big\},~\forall Q\in\mathcal{Q},\label{eqn:equivBellman1}
\end{align}
where $V(Q)=\mathbb{E}[\hat{V}(\mathbf{X})|Q]$.
Then, by the optimal link selection and transmission rate control structure in \eqref{eqn:reducealpha} and \eqref{eqn:rate}, the relationship between $Q'$ and $Q$ via \eqref{eqn:queue-dyn-t} and the per-stage reward $r(\mathbf{X},\Omega(\mathbf{X}))$ in \eqref{eqn:Rpi}, we have \eqref{eqn:reduceAction}. We complete the proof.

\section*{Appendix B: Proof of Lemma~\ref{lemma:propertiesofV}}\label{app:lemma1}
We prove the three properties in Lemma~\ref{lemma:propertiesofV} using RVIA and mathematical induction.

First, we introduce RVIA \cite{bertsekas}. For each $Q\in\mathcal{Q}$, let $V_n(Q)$ be the value function in the $n$th iteration, where $n=0,1,\cdots$. Define
\begin{subequations}
\begin{align}
&J_{n+1}(Q,a_{r,n})\triangleq\bar{p}_s\bar{p}_rV_n(Q)+p_s\bar{p}_rV_n(\min\{Q+R_s,N_r\})\nonumber\\
&+\bar{p}_sp_r\left(\min\{Q,R_r\}+V_n([Q-R_r]^+)\right)\nonumber\\
&+p_sp_r\big[a_{r,n} \min\{Q,R_r\}\nonumber\\
&+V_n\big(Q+\bar{a}_{r,n} \min\{N_r-Q,R_s\}-a_{r,n} \min\{Q,R_r\}\big)\big]\label{eqn:jn1_1}\\
&=\bar{p}_s\bar{p}_rV_n(Q)+p_s\bar{p}_rV_n(\min\{Q+R_s,N_r\})\nonumber\\
&+\bar{p}_sp_r\left(\min\{Q,R_r\}+V_n([Q-R_r]^+)\right)\nonumber\\
&+p_sp_r\big[\mathbf{1}(a_{r,n}=0)V_n(\min\{Q+R_s,N_r\})\nonumber\\
&+\mathbf{1}(a_{r,n}=1)\left(\min\{Q,R_r\}+V_n([Q-R_r]^+)\right)\big],\label{eqn:jn1_2}
\end{align}
\end{subequations}
where $\mathbf{1}(\cdot)$ denotes the indicator function.
Note that $J_{n+1}(Q,a_{r,n})$ is related to the R.H.S of the Bellman equation in \eqref{eqn:reduceAction}. We refer to $J_{n+1}(Q,a_{r,n})$ as the state-action reward function in the $n$th iteration\cite{ngo}.
For each $Q$, RVIA calculates $V_{n+1}(Q)$ as,
\begin{align}\label{eqn:RVIA}
  &V_{n+1}(Q)=\max_{a_{r,n}} J_{n+1}(Q,a_{r,n})-\max_{a_{r,n}} J_{n+1}(Q_0,a_{r,n}),~\forall n
\end{align}
 where $J_{n+1}(Q,a_{r,n})$ is given by \eqref{eqn:jn1_2} and $Q_0\in\mathcal{Q}$ is some fixed state. Under any initialization of $V_0(Q)$, the generated sequence $\{V_n(Q)\}$ converges to $V(Q)$ \cite{bertsekas}, i.e.,
 \begin{equation}
   \lim_{n\to\infty}V_n(Q)=V(Q),~\forall Q\in\mathcal{Q}.\label{eqn:converge}
 \end{equation}
 where $V(Q)$ satisfies the Bellman equation in \eqref{eqn:reduceAction}

 In the following proof, we set $V_0(Q)=0$ for all $Q\in\mathcal{Q}$.
Let $\alpha_{r,n}^*(Q)$ denote the control that attains the maximum of the first term in \eqref{eqn:RVIA} in the $n$th iteration for all $Q$, i.e.,
\begin{equation}
  \alpha_{r,n}^*(Q)=\arg\max_{a_{r,n}}J_{n+1}(Q,a_{r,n}),~~\forall Q\in\mathcal{Q}.\label{eqn:optimaln}
\end{equation}
We refer to $\alpha_{r,n}^*(Q)$ as the optimal policy for the $n$th iteration.
For ease of notation, in the following, we denote $\big(\alpha_{r,n}^*(Q+R_s+R_r+1),\alpha_{r,n}^*(Q+R_s+R_r),\alpha_{r,n}^*(Q+1), \alpha_{r,n}^*(Q)\big)$ as $\big(\alpha_{4,n}^*(Q),\alpha_{3,n}^*(Q),\alpha_{2,n}^*(Q),\alpha_{1,n}^*(Q)\big)$, where $Q\in\{0,1,...,N_r-(R_s+R_r+1)\}$.

Next, we prove Lemma~\ref{lemma:propertiesofV} through mathematical induction using RVIA.

(1) We prove Property 1 by showing that for all $n=0,1,\cdots$, $V_n(Q)$ satisfies
\begin{equation}
  V_n(Q+1)\geq V_n(Q),~Q\in\{0,1,\cdots,N_r-1\}.\label{eqn:p1}
\end{equation}
We initialize $V_0(Q)=0$, for all $Q\in\mathcal{Q}$. Thus, we have $V_0(Q+1)-V_0(Q)=0$, i.e., \eqref{eqn:p1} holds for $n=0$.
Assume that \eqref{eqn:p1} holds for some $n>0$. We will prove that \eqref{eqn:p1} also holds for $n+1$. By \eqref{eqn:RVIA}, we have
\begin{align}
&V_{n+1}(Q+1)=J_{n+1}\left(Q+1,\alpha_{2,n}^*(Q)\right)-\max_{a_{r,n}} J_{n+1}(Q_0,a_{r,n})\nonumber\\
&\overset{(a)}{\geq}J_{n+1}\left(Q+1,\alpha_{1,n}^*(Q)\right)-\max_{a_{r,n}} J_{n+1}(Q_0,a_{r,n})\nonumber\\
&\overset{(b)}{=}\bar{p}_s\bar{p}_rV_n(Q+1)+p_s\bar{p}_rV_n(\min\{Q+1+R_s,N_r\})\nonumber\\
&+\bar{p}_sp_r\left(\min\{Q+1,R_r\}+V_n([Q+1-R_r]^+)\right)\nonumber\\
&+p_sp_r\big[\mathbf{1}\left(\alpha_{1,n}^*(Q)=0\right)V_n(\min\{Q+1+R_s,N_r\})\nonumber\\
&+\mathbf{1}\left(\alpha_{1,n}^*(Q)=1\right)\left(\min\{Q+1,R_r\}+V_n([Q+1-R_r]^+)\right)\big]\nonumber\\
&-\max_{a_{r,n}} J_{n+1}(Q_0,a_{r,n}),\label{eqn:vn1q1}
\end{align}
where $(a)$ follows from the optimality of $\alpha_{2,n}^*(Q)$ for $Q+1$ in the $n$th iteration and  $(b)$ directly follows from \eqref{eqn:jn1_2}. By \eqref{eqn:jn1_2} and \eqref{eqn:RVIA}, we also have
\begin{align}
&V_{n+1}(Q)=J_{n+1}\left(Q,\alpha_{1,n}^*(Q)\right)-\max_{a_{r,n}} J_{n+1}(Q_0,a_{r,n})\nonumber\\
&=\bar{p}_s\bar{p}_rV_n(Q)+p_s\bar{p}_rV_n(\min\{Q+R_s,N_r\})\nonumber\\
&+\bar{p}_sp_r\left(\min\{Q,R_r\}+V_n([Q-R_r]^+)\right)\nonumber\\
&+p_sp_r\big[\mathbf{1}\left(\alpha_{1,n}^*(Q)=0\right)V_n(\min\{Q+R_s,N_r\})\nonumber\\
&+\mathbf{1}\left(\alpha_{1,n}^*(Q)=1\right)\left(\min\{Q,R_r\}+V_n([Q-R_r]^+)\right)\big]\nonumber\\
&-\max_{a_{r,n}} J_{n+1}(Q_0,a_{r,n}).\label{eqn:vn1q}
\end{align}
Next, we compare \eqref{eqn:vn1q1} and \eqref{eqn:vn1q} term by term.
By the facts that $\min\{Q+1+R_s,N_r\}\geq\min\{Q+R_s,N_r\}$, $[Q+1-R_r]^+\geq [Q-R_r]^+$ and $\min\{Q+1,R_r\}\geq\min\{Q,R_r\}$, and the induction hypothesis, we have $V_{n+1}(Q+1)\geq V_n(Q)$, i.e., \eqref{eqn:p1} holds for $n+1$. Therefore, by induction, \eqref{eqn:p1} holds for any $n$.  By taking limits on both sides of \eqref{eqn:p1} and by \eqref{eqn:converge}, we complete the proof of Property 1.

(2) We prove Property 2 by showing that for all $n=0,1,\cdots$, $V_n(Q)$ satisfies
\begin{equation}
  V_n(Q+1)- V_n(Q)\leq 1,~Q\in\{0,1,\cdots,N_r-1\}.\label{eqn:p2}
\end{equation}
We initialize $V_0(Q)=0$, for all $Q\in\mathcal{Q}$. Thus, we have $V_0(Q+1)-V_0(Q)=0$, i.e., \eqref{eqn:p2} holds for $n=0$.
Assume that \eqref{eqn:p2} holds for some $n>0$. We will prove that \eqref{eqn:p2} also holds for $n+1$. By \eqref{eqn:RVIA} and \eqref{eqn:jn1_2}, we have,
\begin{align}
  &V_{n+1}(Q+1)-V_{n+1}(Q)\nonumber\\
  =&J_{n+1}\left(Q+1,\alpha_{2,n}^*(Q)\right)-J_{n+1}\left(Q,\alpha_{1,n}^*(Q)\right)\nonumber\\
  =&\left[J_{n+1}\left(Q+1,\alpha_{2,n}^*(Q)\right)-J_{n+1}\left(Q,\alpha_{2,n}^*(Q)\right)\right]\nonumber\\
  &~~+\left[J_{n+1}\left(Q,\alpha_{2,n}^*(Q)\right)-J_{n+1}\left(Q,\alpha_{1,n}^*(Q)\right)\right]\nonumber\\
  \overset{(c)}{\leq}&J_{n+1}\left(Q+1,\alpha_{2,n}^*(Q)\right)-J_{n+1}\left(Q,\alpha_{2,n}^*(Q)\right)\nonumber\\
  =&\bar{p}_s\bar{p}_rA_1+p_s\bar{p}_rB_1+\bar{p}_sp_rC_1+p_sp_rD_1,\label{eqn:property2deltaJ}
\end{align}
where (c) is due to $J_{n+1}(Q,\alpha_{2,n}^*(Q))\leq J_{n+1}(Q,\alpha_{1,n}^*(Q))$. This is because $\alpha_{1,n}^*(Q)$ is the optimal policy for $Q$ in the $n$th iteration.
$A_1,B_1,C_1$ and $D_1$ in  \eqref{eqn:property2deltaJ} are given as follows.
\begin{subequations}
\begin{align}
  &A_1= V_n(Q+1)-V_n(Q),\label{eqn:a1}\\
  &B_1= V_n(\min\{Q+1+R_s,N_r\})-V_n(\min\{Q+R_s,N_r\}),\label{eqn:b1}\\
  &C_1= \min\{Q+1,R_r\}+V_n([Q+1-R_r]^+)\nonumber\\&~~~~~-\min\{Q,R_r\}-V_n([Q-R_r]^+),\label{eqn:c1}\\
  &D_1= \mathbf{1}\left(\alpha_{2,n}^*(Q)=0\right)B_1+\mathbf{1}\left(\alpha_{2,n}^*(Q)=1\right)C_1.\label{eqn:d1}
\end{align}
\end{subequations}
Note that $p_s+\bar{p}_s=1$ and $p_r+\bar{p}_r=1$. Thus, to show $V_{n+1}(Q+1)-V_{n+1}(Q)\leq 1$ using \eqref{eqn:property2deltaJ}, it suffices to show that $A_1\leq 1$, $B_1\leq 1$, $C_1\leq 1$ and $D_1\leq 1$.
Due to the induction hypothesis, $A_1\leq 1$ and $B_1\leq 1$ hold.
To prove $C_1\leq 1$, we consider the following two cases.
(i) When $Q\geq R_r$, we have $C_1=V_n(Q+1-R_r)-V_n(Q-R_r)\leq 1$ due to the induction hypothesis.
(ii) When $Q\leq R_r-1$, we have $C_1=1$.
Thus $C_1\leq 1$ holds.
To prove $D_1\leq 1$, we consider the following two cases.
(i) If $\alpha_{2,n}^*(Q)=0$, we have $D_1=B_1\leq 1$.
(ii) If $\alpha_{2,n}^*(Q)=1$, we have $D_1=C_1\leq 1$.  Thus, we can show that \eqref{eqn:p2} holds for $n+1$. Therefore, by induction \eqref{eqn:p2} holds for any $n$.  By taking limits on both sides of \eqref{eqn:p2} and by \eqref{eqn:converge}, we complete the proof of Property 2.

(3) We prove Property 3 by showing that for all $n=0,1,\cdots$, $V_n(Q)$ satisfies
\begin{align}
  &V_n(Q+R_s+R_r+1)-V_n(Q+R_s+R_r)\leq V_n(Q+1)\nonumber\\
  &-V_n(Q),~Q\in\{0,1,...,N_r-(R_s+R_r+1)\}.\label{eqn:p3}
\end{align}
We initialize $V_0(Q)=0$, for all $Q\in\mathcal{Q}$. Thus, we have $V_0(Q+R_s+R_r+1)-V_0(Q+R_s+R_r)=V_0(Q+1)-V_0(Q)=0$, i.e., \eqref{eqn:p3} holds for $n=0$.
Assume that \eqref{eqn:p3} holds for some $n>0$. We will prove that \eqref{eqn:p3} also holds for $n+1$. By \eqref{eqn:RVIA}, we have,
\begin{align}
  &V_{n+1}(Q+R_s+R_r+1)-V_{n+1}(Q+R_s+R_r)\nonumber\\
  =&J_{n+1}\left(Q+R_s+R_r+1,\alpha_{4,n}^*(Q)\right)\nonumber\\
  &-J_{n+1}\left(Q+R_s+R_r,\alpha_{3,n}^*(Q)\right)\nonumber\\
  =&\big[J_{n+1}\left(Q+R_s+R_r+1,\alpha_{4,n}^*(Q)\right)\nonumber\\
  &-J_{n+1}\left(Q+R_s+R_r,\alpha_{4,n}^*(Q)\right)\big]\nonumber\\
  &+\big[J_{n+1}\left(Q+R_s+R_r,\alpha_{4,n}^*(Q)\right)\nonumber\\
  &-J_{n+1}\left(Q+R_s+R_r,\alpha_{3,n}^*(Q)\right)\big]\nonumber\\
  \overset{(d)}{\leq}&J_{n+1}\left(Q+R_s+R_r+1,\alpha_{4,n}^*(Q)\right)\nonumber\\
  &-J_{n+1}\left(Q+R_s+R_r,\alpha_{4,n}^*(Q)\right),\label{eqn:rhsp31}
\end{align}
and
\begin{align}
  &V_{n+1}(Q+1)-V_{n+1}(Q)\nonumber\\
  =&J_{n+1}\left(Q+1,\alpha_{2,n}^*(Q)\right)-J_{n+1}\left(Q,\alpha_{1,n}^*(Q)\right)\nonumber\\
  =&\left[J_{n+1}\left(Q+1,\alpha_{2,n}^*(Q)\right)-J_{n+1}\left(Q+1,\alpha_{1,n}^*(Q)\right)\right]\nonumber\\
  &+\left[J_{n+1}\left(Q+1,\alpha_{1,n}^*(Q)\right)-J_{n+1}\left(Q,\alpha_{1,n}^*(Q)\right)\right]\nonumber\\
  \overset{(e)}{\geq}&J_{n+1}\left(Q+1,\alpha_{1,n}^*(Q)\right)-J_{n+1}\left(Q,\alpha_{1,n}^*(Q)\right),\label{eqn:rhsp32}
\end{align}
where $(d)$ and $(e)$ are due to $J_{n+1}(Q+R_s+R_r,\alpha_{4,n}^*(Q))\leq J_{n+1}(Q+R_s+R_r,\alpha_{3,n}^*(Q))$ and $J_{n+1}(Q+1,\alpha_{2,n}^*(Q))\geq J_{n+1}(Q+1,\alpha_{1,n}^*(Q))$, respectively. This is because $\alpha_{3,n}^*(Q)$ and $\alpha_{2,n}^*(Q)$ are the optimal policies for  $Q+R_s+R_r$ and $Q+1$ in the $n$th iteration, respectively.

To show that $V_{n+1}(Q+R_s+R_r+1)-V_{n+1}(Q+R_s+R_r)\leq V_{n+1}(Q+1)-V_{n+1}(Q)$, it suffices to show that $J_{n+1}(Q+R_s+R_r+1,\alpha_{4,n}^*(Q))-J_{n+1}(Q+R_s+R_r,\alpha_{4,n}^*(Q))\leq J_{n+1}(Q+1,\alpha_{1,n}^*(Q))-J_{n+1}(Q,\alpha_{1,n}^*(Q))$, i.e., the R.H.S. of \eqref{eqn:rhsp31} is no greater than the R.H.S. of \eqref{eqn:rhsp32}.
By \eqref{eqn:jn1_2}, we have
\begin{align}
  &J_{n+1}\left(Q+R_s+R_r+1,\alpha_{4,n}^*(Q)\right)\nonumber\\&-J_{n+1}\left(Q+R_s+R_r,\alpha_{4,n}^*(Q)\right)\nonumber\\
  =&\bar{p}_s\bar{p}_rA_2+p_s\bar{p}_rB_2+\bar{p}_sp_rC_2+p_sp_rD_2,\label{eqn:property3deltaJ11}
\end{align}
where
\begin{subequations}
\begin{align}
  A_2=&V_n(Q+R_s+R_r+1)-V_n(Q+R_s+R_r),\label{eqn:a2}\\
  B_2=&V_n(\min\{Q+2R_s+R_r+1,N_r\})\nonumber\\
&-V_n(\min\{Q+2R_s+R_r,N_r\}),\label{eqn:b2}\\
  C_2=&V_n(Q+R_s+1)-V_n(Q+R_s),\label{eqn:c2}\\
  D_2=&\mathbf{1}\left(\alpha_{4,n}^*(Q)=0\right)B_2+\mathbf{1}\left(\alpha_{4,n}^*(Q)=1\right)C_2,\label{eqn:d2}
\end{align}
\end{subequations}
and
\begin{align}
&J_{n+1}\left(Q+1,\alpha_{1,n}^*(Q)\right)-J_{n+1}\left(Q,\alpha_{1,n}^*(Q)\right)\nonumber\\
=&\bar{p}_s\bar{p}_rA_1+p_s\bar{p}_rB_1+\bar{p}_sp_rC_1+p_sp_rD_1',\label{eqn:property3deltaJ12}
\end{align}
where
$A_1$, $B_1$ and $C_1$ are given by \eqref{eqn:a1}, \eqref{eqn:b1} and \eqref{eqn:c1}, respectively,
and
\begin{equation*}
D_1'=\mathbf{1}\left(\alpha_{1,n}^*(Q)=0\right)B_1+\mathbf{1}\left(\alpha_{1,n}^*(Q)=1\right)C_1.\label{eqn:d11}
\end{equation*}
Note that, when $Q\in\{0,1,...,N_r-(Q+R_s+R_r+1)\}$, \eqref{eqn:b1} can be rewritten as $B_1=V_n(Q+1+R_s)-V_n(Q+R_s)$.

To show that \eqref{eqn:p3} holds for $n+1$ using \eqref{eqn:property3deltaJ11} and \eqref{eqn:property3deltaJ12}, it suffices to show that $A_2\leq A_1$, $B_2\leq B_1$, $C_2\leq C_1$ and $D_2\leq D_1'$. Due to the induction hypothesis, $A_2\leq A_1$ holds.
To prove $B_2\leq B_1$, we consider the following two cases. (i) When $Q+2R_s+R_r\geq N_r$, we have $B_2-B_1=V_n(Q+R_s)-V_n(Q+R_s+1)\leq 0$  as \eqref{eqn:p1} holds for any $n$. (ii) When $Q+2R_s+R_r+1\leq N_r$, we have $B_2-B_1=V_n(Q+2R_s+R_r+1)-V_n(Q+2R_s+R_r)-\left(V_n(Q+R_s+1)-V_n(Q+R_s)\right)\leq 0$ due to the induction hypothesis.
Thus, $B_2\leq B_1$ holds.
To prove $C_2\leq C_1$, we consider two cases. (i) When $Q\leq R_r-1$, we have $C_2-C_1=V_n(Q+R_s+1)-V_n(Q+R_s)-1\leq 0$ as \eqref{eqn:p2} holds for any $n$. (ii) When $Q\geq R_r$, we have $C_2-C_1=V_n(Q+R_s+1)-V_n(Q+R_s)-\left(V_n(Q+1-R_r)-V_n(Q-R_r)\right)\leq 0$ due to the induction hypothesis.
Thus, $C_2\leq C_1$ holds.
To prove $D_2\leq D_1'$, we consider the following four cases. (i) If $\left(\alpha_{4,n}^*(Q),\alpha_{1,n}^*(Q)\right)=(0,0)$, we have $D_2-D_1'=B_2-B_1\leq 0$. (ii) If $\left(\alpha_{4,n}^*(Q),\alpha_{1,n}^*(Q)\right)=(1,0)$, we have $D_2-D_1'=C_2-B_1=0$. (iii) If $\left(\alpha_{4,n}^*(Q),\alpha_{1,n}^*(Q)\right)=(1,1)$, we have $D_2-D_1'=C_2-C_1\leq 0$. (iv) If $\left(\alpha_{4,n}^*(Q),\alpha_{1,n}^*(Q)\right)=(0,1)$, we have $D_2-D_1'=B_2-C_1\leq B_1-C_1=C_2-C_1\leq 0$. Thus, $D_2\leq D_1'$ holds.

We have shown that $A_2\leq A_1$, $B_2\leq B_1$, $C_2\leq C_1$ and $D_2\leq D_1'$. Thus, \eqref{eqn:p3} holds for $n+1$. Therefore, by induction \eqref{eqn:p3} holds for any $n$.
By taking limits on both sides of \eqref{eqn:p3} and by \eqref{eqn:converge}, we complete the proof of Property 3.

\section*{Appendix C: Proof of Theorem~\ref{theorem:theorem2}}\label{app:theorem2}
First, we show that $J(Q,a_r)$ in \eqref{eqn:state_action_func} is supermodular in $(Q,a_r)$.
By the definition of supermodularity, it is equivalent to prove $\Delta J(Q+1)-\Delta J(Q)\geq 0$, where $\Delta J(Q)\triangleq \left(J(Q,1)-J(Q,0)\right)/p_sp_r$.
By \eqref{eqn:state_action_func}, we have
\begin{eqnarray}
  \Delta J(Q+1)-\Delta J(Q)=\min\{Q+1, R_r\}-\min\{Q, R_r\}&&\nonumber\\
  +V([Q+1-R_r]^+)-V([Q-R_r]^+)~~~~~~~~~~~~~~~~~~~~&&\nonumber\\
 -V(\min\{Q+1+R_s, N_r\})+V(\min\{Q+R_s, N_r\}). ~&&
  \label{eqn:diff_deltaJ}
\end{eqnarray}

To prove $\Delta J(Q+1)-\Delta J(Q)\geq 0$ using \eqref{eqn:diff_deltaJ}, we consider the following two cases.

(1) If $N_r\leq R_s+R_r$, we consider three cases.
  (i) When $Q\geq R_r$, then $\Delta J(Q+1)-\Delta J(Q)=V(Q-R_r+1)-V(Q-R_r)\geq 0$.
  (ii) When $N_r-R_s\leq Q \leq R_r-1$, then $\Delta J(Q+1)-\Delta J(Q)=1$.
  (iii) When $Q\leq N_r-R_s-1$, then $\Delta J(Q+1)-\Delta J(Q)=1-(V(Q+R_s+1)-V(Q+R_s))\geq 0$.

(2) If $N_r\geq R_s+R_r+1$, we consider the similar three cases.
  (i) When $Q\geq N_r-R_s$, then $\Delta J(Q+1)-\Delta J(Q)=V(Q-R_r+1)-V(Q-R_r)\geq 0$.
  (ii) When $R_r\leq Q \leq N_r-R_s-1$, then $\Delta J(Q+1)-\Delta J(Q)=\big(V(Q+1-R_r)-V(Q_r-R_r)\big)-\big(V(Q+R_s+1)-V(Q+R_s)\big)\geq 0$.
  (iii) When $Q\leq R_r-1$, then $\Delta J(Q+1)-\Delta J(Q)=1-(V(Q+R_s+1)-V(Q+R_s))\geq 0$.

Therefore, $\Delta J(Q+1)-\Delta J(Q)\geq 0$ holds which implies that  $J(Q,a_r)$ in \eqref{eqn:state_action_func} is supermodular in $(Q, a_r)$. According to \cite[Lemma 4.7.1]{puterman}, the optimal policy $\alpha_r^*(Q)$ given by \eqref{eqn:omegaalpha} is monotonically non-decreasing in $Q$. Thus $\alpha_r^*(Q)$ has the threshold-based structure in \eqref{eqn:threshold} which completes the proof.

\section*{Appendix D: Proof of Lemma~\ref{lemma:recurrentclass}}
 Consider $l=0$, i.e., $N_r=nR$. Assume $Q_1=0$ is the initial state. According the B\'{e}zout's identity and the queue dynamics in \eqref{eqn:queue-dyn-t}, for all $t=1,2,\cdots$ and $k=1,2,\cdots,n$, there exist integers $x_{t,k}$ and $y_{t,k}$ such that $Q_t=Q_1+x_{t,k}a+y_{t,k}b=kR$.
Denote $\mathcal{S}=\{R,2R,\cdots,nR\}$. In other words, each state $s\in\mathcal{S}$ is accessible from state $0$.
On the other hand, for any initial state $Q_1\in\mathcal{Q}$, we have $\Pr[Q_t=0]>0$ for some $t$.
The reason is that CSI may stay $(0,1)$ for enough consecutive time slots which implies that the relay buffer will be empty under any policy in Definition~\ref{definition:definition1}.
Thus, state $0$ is  accessible from all states in $\mathcal{Q}$. Note that $\mathcal{S}\subseteq\mathcal{Q}$.
Therefore, by \cite[Definition 4.2.5]{gallager}, state $0$ is a recurrent state. Denote $\mathcal{C}=\{0\}\cup \mathcal{S}$. Then, by \cite[Theorem 4.2.1]{gallager}, $\mathcal{C}$ is a recurrent class.
Note that, under the optimal transmission rate policy in \eqref{eqn:rate}, each state $\not\in\mathcal{C}$ is not accessible from state $0$.
Thus, by \cite[Definition 4.2.5]{gallager}, the states $\not\in\mathcal{C}$ are all transient states.
Therefore, if $l=0$, the recurrent class of $\{Q_t\}$ is $\mathcal{C}=\{0,R,2R,\cdots,nR\}$.

Consider $l\neq 0$, i.e., $N_r=nR+l$. $\mathcal{C}_1\triangleq\{0,R,2R,\cdots,nR\}$ is still a recurrent class. Assume $Q_1=N_r$ is the initial state. Similarly, for all $t=1,2,\cdots$ and $k=0,1,\cdots,n-1$, there exist integers $x_{t,k}$ and $y_{t,k}$ such that $Q_t=Q_1+x_{t,k}a+y_{t,k}b=kR+l$.
Denote $\mathcal{S}'=\{l,l+R,\cdots,l+(n-1)R\}$.
Then, each state $s\in\mathcal{S}'$ is accessible from state $N_r$.
On the other hand, for any initial state $Q_1\in\mathcal{Q}$, we have $\Pr[Q_t=N_r]>0$ for some $t$.
The reason is that CSI may stay $(1,0)$ for enough consecutive time slots  which implies that the buffer will be full under any policy in Definition~\ref{definition:definition1}.
Thus, state $N_r$ is accessible from all states in $\mathcal{Q}$.
Note that $\mathcal{S}'\subseteq\mathcal{Q}$.
Therefore, state $N_r$ is a recurrent state and $\mathcal{C}_2\triangleq \{N_r\}\cup \mathcal{S}'$ is a recurrent class\cite{gallager}.
Note that, state $0$ and $N_r$ are accessible from each other.
Thus, $\mathcal{C}\triangleq\mathcal{C}_1 \cup \mathcal{C}_2$ is a recurrent class.
Similarly, by \eqref{eqn:rate}, the states $\not\in\mathcal{C}$ are all transient states.
Therefore, if $l\neq 0$, the recurrent class of $\{Q_t\}$ is $\mathcal{C}=\{0,R,2R,\cdots,nR,l,l+R,l+R,\cdots,N_r\}$.
We complete the proof.

\section*{APPENDIX E: Proof of Lemma~\ref{lemma:relationship12}}\label{app:lemma3}
By Lemma~\ref{lemma:recurrentclass}, for any given $R_s, R_r, N_r$, the recurrent class $\mathcal{C}$ is fixed  and does not change with the threshold $Q_{th}\in\mathcal{Q}$, and the ergodic system throughput in \eqref{eqn:problem2} only depends on the steady-state probability vector $\bm{\pi}(Q_{th})$ and the average departure rate vector $\mathbf{r}(Q_{th})$ of $\mathcal{C}$.
Note that, by \eqref{eqn:comp-pi} and \eqref{eqn:comp-ri}, $\bm{\pi}(Q_{th})$ and $\mathbf{r}(Q_{th})$ only depend on the link selection control for the recurrent states in $\mathcal{C}$.
Since $q_{th}^*$ and $q_{th,next}^*$  are two adjacent recurrent states, by \eqref{eqn:thresholdlemma2}, any threshold $Q_{th}^*\in\{Q|q_{th}^*\leq Q<q_{th,next}^*,Q\in\mathcal{Q}\}$ leads to the same link selection control for the recurrent states, and hence achieves the same ergodic throughput.  Under the stationary unichain policies in \eqref{eqn:rate} and \eqref{eqn:thresholdlemma2}, the induced markov chain $\{Q_t\}$ is an ergodic unichain. By ergodic theory, the time-average system throughput in \eqref{eqn:problem1} is equivalent to the ergodic system throughput in \eqref{eqn:problem2}, i.e., $\bar{R}^*=\bar{r}^*$. Thus, the optimal control to Problem~\ref{problem:originalproblem} can be obtained by solving Problem~\ref{problem:equivalentproblem}. We complete the proof of Lemma~\ref{lemma:relationship12}.

\section*{APPENDIX F: Proof of Lemma~\ref{lemma:lemmaofupdate}}\label{app:lemma4}
  For two adjacent thresholds $c_k$ and $c_{k+1}$, the corresponding transition probability matrices  $\mathbf{P}(k)$ and $\mathbf{P}(k+1)$ only differ in the $(k+1)$-th row. Thus, $\mathbf{A}(k)$ and $\mathbf{A}(k+1)$ only differ in the $(k+1)$-th column. Then, by partitioning $\mathbf{A}(k)$ and $\mathbf{A}(k+1)$ into the form \eqref{eqn:permu} using the permutation matrices $\mathbf{K}(k)$ and $\mathbf{K}(k+1)$, respectively, we obtain the corresponding submatrices $\hat{\mathbf{A}}(k)$ and $\hat{\mathbf{A}}(k+1)$. By exchanging the $(k+1)$-th and $(k+2)$-th columns of $\hat{\mathbf{A}}(k)$, we obtain $\tilde{\mathbf{A}}(k)$, i.e.,
  \begin{equation}
\tilde{\mathbf{A}}(k)=\hat{\mathbf{A}}(k)\hat{\mathbf{K}}(k),\label{eqn:defi-hatakk}
\end{equation}
where $\hat{\mathbf{K}}(k)$ is the corresponding permutation matrix defined in Lemma~\ref{lemma:lemmaofupdate}.
Thus, $\tilde{\mathbf{A}}(k)$ and $\hat{\mathbf{A}}(k+1)$ differ only in the $(k+1)$-th column, and $\hat{\mathbf{A}}(k+1)$ can be regarded as a rank-one update of $\tilde{\mathbf{A}}(k)$. Let
\begin{equation}
    \mathbf{u}(k)=\mathbf{a}_{k+1}(k+1)-\hat{\mathbf{a}}_{k+1}(k),\label{eqn:defiukproof}
\end{equation}
where $\mathbf{a}_{k+1}(k+1)$ and $\hat{\mathbf{a}}_{k+1}(k)$ are the $(k+1)$-column of $\hat{\mathbf{A}}(k+1)$ and $\tilde{\mathbf{A}}(k)$, respectively.
Then, we have $\hat{\mathbf{A}}(k+1)=\tilde{\mathbf{A}}(k)+\mathbf{u}(k)\mathbf{v}(k)^T,$
where $\mathbf{v}(k)$ is defined in \eqref{eqn:defivk}.
By the Sherman-Morrison formula\cite{updating}, we have
\begin{align}
  \hat{\mathbf{A}}(k+1)^{-1}&=\tilde{\mathbf{A}}(k)^{-1}
  -\frac{\tilde{\mathbf{A}}(k)^{-1}\mathbf{u}(k)\mathbf{v}(k)^T\tilde{\mathbf{A}}(k)^{-1}}{1+\mathbf{v}(k)^T\tilde{\mathbf{A}}(k)^{-1}\mathbf{u}(k)}.\label{eqn:comp-nextk2}
\end{align}
By \eqref{eqn:defi-hatakk}, we have $\tilde{\mathbf{A}}(k)^{-1}=\hat{\mathbf{K}}(k)\hat{\mathbf{A}}(k)^{-1}$ and $\hat{\mathbf{a}}_{k+1}(k)=\mathbf{a}_{k+2}(k)$. Thus, \eqref{eqn:defiukproof} is equivalent to \eqref{eqn:defiuk} and \eqref{eqn:comp-nextk2} is equivalent to \eqref{eqn:comp-nextk}. We complete the proof.

\section*{APPENDIX G: Proof of Lemma~\ref{lemma:solution2Symm}}\label{app:lemma5}

First, we show that Problem~\ref{problem:equivalentproblem} can be equivalently transformed to Problem~\ref{problem:problemSymm}.
Given a CSI sample path $\{\mathbf{G}_t\}$, let $\{(a_{s,t},a_{r,t})\}$ and $\{(a_{s,t},a_{r,t})\}$ be the sequences of link selection and transmission rate actions under a policy $\Omega$ in Definition~\ref{definition:definition1}, respectively.
Let $\{Q_t\}$ be the associated QSI trajectory.
By \eqref{eqn:queue-dyn-t}, we have  $\frac{1}{T}\sum_{t=1}^T a_{s,t}u_{s,t}=\frac{1}{T}\sum_{t=1}^T a_{r,t}u_{r,t}+\frac{Q_T-Q_1}{T}$ for all $T$.
Since $Q_1, Q_T\leq N_r$, by taking expectation over all sample paths and $\liminf$, we have
$\liminf_{T\to\infty}\frac{1}{T}\sum_{t=1}^T \mathbb{E}\left[a_{s,t}u_{s,t}\right]=\liminf_{T\to\infty}\frac{1}{T}\sum_{t=1}^T \mathbb{E}\left[a_{r,t}u_{r,t}\right]=\bar{R}^\Omega$.
Hence, under the optimal transmission rate control in \eqref{eqn:rate} and a threshold-based link selection policy in \eqref{eqn:thresholdlemma2}, the average arrival rate equals to the average departure rate.
Therefore, without loss of optimality, we can also denote $r_i(q_{th})$ as the average of the average arrival and departure rates at state  $i\in\mathcal Q$ under the threshold $q_{th}$.
Then, in the symmetric case, we have
\begin{equation}
  r_i(m)=\left\{
           \begin{array}{ll}
                ~~~\frac{pR}{2}, & \hbox{$i=0,n$;} \\
                    \frac{(p+p\bar{p})R}{2}, & \hbox{$i=1,...,n-1$.}
           \end{array}
         \right.\label{eqn:reward}
\end{equation}
where state $i\in\{0,1,\cdots,n\}$ represents state $iR\in\mathcal{C}$ and $m=q_{th}/R$.
By \eqref{eqn:reward}, we obtain the ergodic system throughput $\bar{r}(m)=\frac{(p+p\bar{p})R}{2}-\frac{p\bar{p}R}{2}(\pi_0(m)+\pi_n(m))$. Maximizing $\bar{r}(m)$ is equivalent to minimizing $\pi_0(m)+\pi_n(m)$. Therefore, by \eqref{eqn:pi0pin}, we have the discrete optimization problem in Problem~\ref{problem:problemSymm}.

Next, we obtain the optimal solution to Problem~\ref{problem:problemSymm}.
By change of variables in Problem~\ref{problem:problemSymm}, i.e., letting $x=\bar{p}^m$, we have the following continuous optimization problem.
\begin{align}
\min_{x\in[\bar{p}^n, 1]} &~~~g(x)=\frac{\bar{p}^{n+1}-\bar{p}^n+\bar{p}^2x^2-\bar{p}x^2}{\bar{p}^2x^2+\bar{p}^{n+1}-2\bar{p}x}.\label{eqn:optimalx}
\end{align}
Letting the derivative of $g(x)$, i.e., $g'(x)=\frac{2\bar{p}p^2(x^2-\bar{p}^{n-1})}{(\bar{p}^2x^2+\bar{p}^{n+1}-2\bar{p}x)^2}$, equal to $0$, we have $x^*=\bar{p}^{\frac{n-1}{2}}$. Since $g'(x)\leq 0$ in $[\bar{p}^n,\bar{p}^{\frac{n-1}{2}}]$ and $g'(x)\geq 0$ in $[\bar{p}^{\frac{n-1}{2}},1]$, $x^*=\bar{p}^{\frac{n-1}{2}}$ is the optimal solution to \eqref{eqn:optimalx}.
Based on $x^*$, we now obtain the optimal solution $m^*$ to Problem~\ref{problem:problemSymm}.
If $n$ is odd, then $m^*=\frac{n-1}{2}$; if $n$ is even, $m^*=\frac{n}{2}-1$ or $\frac{n}{2}$. (Note that, $\frac{n}{2}-1$ and $\frac{n}{2}$ achieve the same optimal value of Problem~\ref{problem:problemSymm}.)
Then, by Lemma~\ref{lemma:relationship12}, we complete the proof.

\bibliographystyle{IEEEtran}
\bibliography{IEEEabrv,revised}

\end{document}